\documentclass[10pt,journal]{IEEEtran}
\IEEEoverridecommandlockouts

\usepackage{graphicx}
\usepackage{amsmath,amsthm,amsfonts,amssymb}
\usepackage{cite,hyperref}
\usepackage{bm}
\usepackage{bbm}
\usepackage{url}
\usepackage{array}
\usepackage{color,soul}
\usepackage{multirow}
\usepackage{booktabs}
\usepackage[table,xcdraw]{xcolor}
\usepackage{enumitem}
\usepackage{subcaption}

\usepackage[english]{babel}
\usepackage{algorithm}
\usepackage{algorithmic}
\usepackage{setspace}
\usepackage[english]{babel}

\theoremstyle{plain}
\newtheorem{thm}{Theorem}
\newtheorem{lem}[thm]{Lemma}

\newtheorem{cor}[thm]{Corollary}

\newtheorem{sty1}{Theorem}
\newtheorem{defi}[sty1]{Definition}

\allowdisplaybreaks[4]

\begin{document}
\title{Multi-Mode Pinching-Antenna Systems: Polarization-Aware Full-Wave Modeling and Optimization}

\author{
Dengke Wei,
Runxin Zhang,
Yulin Shao,
Fen Hou,
Shaodan Ma

\thanks{D. Wei, F. Hou, and S. Ma are with the State Key Laboratory of Internet of Things for Smart City, University of Macau, Macau, China (dengke.wei@connect.um.edu.mo, \{fenhou,~shaodanma\}@um.edu.mo).}
\thanks{R. Zhang is with the Department of Electronic Engineering, Tsinghua University, Beijing, China (rxz@mail.tsinghua.edu.cn)
}
\thanks{Y. Shao is with the Department of Electrical and Electronic Engineering, The University of Hong Kong, Hong Kong, China (ylshao@hku.hk).
}
\thanks{This work was performed during D. Wei's visit at the Department of Electrical and Computer Engineering, The University of Hong Kong.
}
}

\maketitle

\begin{abstract}
Millimeter-wave and terahertz communications face a fundamental challenge: overcoming severe path loss without sacrificing spectral efficiency. Pinching antenna systems (PASS) address this by bringing radiators physically close to users, yet existing frameworks treat the waveguide as a mere transmission line, overlooking its inherent multi-mode capabilities and the critical role of polarization. 
This paper develops the first polarization-aware, full-wave electromagnetic model for multi-mode PASS (MMPASS), capturing spatial radiation patterns, modal polarization states, and polarization matching efficiency from first principles. 
Leveraging this physically grounded model, we reveal fundamental trade-offs among waveguide attenuation, atmospheric absorption, and geometric spreading, yielding closed-form solutions for optimal PA placement and orientation in single-user scenarios. 
Extending to multi-user settings, we propose a modular optimization framework that integrates fractional programming with closed-form polarization updates, scaling gracefully to arbitrary numbers of waveguides, PAs, and users.
Numerical results show that MMPASS achieves up to a $167\%$ increase in spectral efficiency compared with single-mode PASS. Moreover, when comparing MMPASS with its polarization-ignorant counterpart, polarization awareness alone improves the sum rate by up to $23\%$.
By bridging rigorous electromagnetic theory with scalable optimization, MMPASS establishes a physically complete and practically viable foundation for future high-frequency wireless networks.
\end{abstract}


\section{Introduction}
The pursuit of ubiquitous connectivity, ultra-high data rates, and massive device interoperability in the 6G era \cite{dang2020should,saad2019vision,liu2025capa,shao2024theory} has fundamentally challenged conventional antenna architectures\cite{ding2025flexible,liu2025capa,zhang2024polarization,wu2021intelligent}. As we venture into millimeter-wave (mmWave) and terahertz (THz) bands, the abundant bandwidth comes at the cost of severe path loss, atmospheric absorption, and acute sensitivity to blockages\cite{chang2022integrated,gao2022integrated,xue2024survey}. These impairments demand novel paradigms that go beyond traditional multiple-input multiple-output (MIMO) arrays, hybrid beamforming, or even reconfigurable intelligent surfaces (RIS), all of which still rely on fixed radiating elements and suffer from fundamental distance-dependent fading\cite{wu2021intelligent,tang2020wireless,shao2021federated,dai2020reconfigurable}.

In this context, pinching antenna systems (PASS) \cite{suzuki2022pinching,ding2025flexible,xu2026generalized} have emerged as a transformative architecture that synergistically combines low-loss guided-wave propagation with the agility of mobile antennas. First prototyped by NTT DOCOMO \cite{suzuki2022pinching}, PASS deploys a dielectric waveguide as the signal backbone. Along this waveguide, multiple non-contact radiating elements, called pinching antennas (PAs), can be placed arbitrarily \cite{zhang2025directional,wang2025modeling,chen2025dynamic,11212813}. By physically positioning a PA in close proximity to a user, the system dramatically shortens the free-space propagation distance, effectively bypassing the severe path loss and blockage sensitivity that plague conventional wireless links. This unique capability has spurred a growing body of research, covering fundamental performance analysis \cite{ding2025flexible,zhang2025directional,qin2025joint,zhou2026joint,11368709}, array gain characterization \cite{wang2025modeling,11212813}, joint beamforming design \cite{wang2025modeling,11165763,chen2026hybrid}, and even learning-based optimization\cite{11358830,karagiannidis2025deep}. Collectively, these works have firmly established PASS as a promising candidate for future high-frequency wireless networks.

A key observation, however, is that PASS is fundamentally waveguide communication. 
A dielectric waveguide is not merely a low-loss cable; it is an electromagnetic structure that inherently supports multiple guided modes, each being a distinct field solution to Maxwell's equations under given boundary conditions\cite{yeh2008essence}. 
Each mode propagates with its own propagation constant and transverse field distribution, and modes are orthogonal to one another. 
This physical property naturally invites a powerful communication concept: \emph{mode-division multiplexing} (MDM)\cite{luo2014wdm}. 
By exciting multiple orthogonal modes on the same waveguide, one can transmit independent data streams simultaneously over a single physical waveguide, thereby multiplying the system's degrees of freedom and spectral efficiency without deploying additional waveguides\cite{vcivzmar2012exploiting}.

Despite this clear opportunity, the vast majority of existing PASS literature remains focused on single-mode operation, where only the fundamental mode is excited. One notable exception is the parallel and contemporaneous work by Xu et al. \cite{xu2026multimodepinchingantennasystems}, which independently proposes a multi-mode PASS framework, demonstrating that mode selectivity can be exploited for multi-user communications. Using coupled-mode theory (CMT) \cite{haus2002coupled}, they develop optimization algorithms, such as Newton-based search and particle swarm optimization, to jointly design PA locations and baseband beamforming, and show significant gains over single-mode PASS.

However, as a CMT-based black-box model, their approach characterizes power transfer via aggregate coupling coefficients but does not describe the spatial field distribution radiated from each PA \cite{xu2026multimodepinchingantennasystems}. Consequently, it cannot capture the directional radiation patterns, nor can it account for the polarization state of the radiated electric field. Polarization is a critical factor at mmWave and THz frequencies \cite{stutzman2018polarization}, where a mismatch can cause several decibels of additional loss. Moreover, the receiver in \cite{xu2026multimodepinchingantennasystems} is modeled as a scalar isotropic antenna, ignoring the fact that optimal receive polarization must align with the incident field. These simplifications, while tractable, leave a significant performance gap unexploited and limit the physical grounding of the system design.

In this paper, we overcome these limitations by developing the first polarization-aware, full-wave electromagnetic modeling framework for multi-mode PASS (MMPASS). 
Departing from CMT-based black-box approach, our model captures the spatial radiation pattern and, more importantly, the polarization state of each guided mode at the receiver side. This enables us to formulate the end-to-end channel with polarization matching efficiency incorporated via Jones vectors, a physical dimension absent in all prior PASS literature. Building upon this physically grounded model, we address the sum-rate maximization problem for multi-user multi-PASS systems. The main contributions of this paper are summarized as follows.
\begin{itemize}[leftmargin=0.5cm]
    \item We establish the first complete, first-principles electromagnetic model for multi-mode pinching-antenna systems. We derive the modal field distributions inside the rectangular dielectric waveguide, compute the radiated fields from each PA port using the equivalence principle, and characterize the polarization state of each mode. This model provides a transparent physical foundation that goes far beyond the black-box CMT description.
    \item We introduce polarization matching efficiency into the PASS channel, and derive closed-form expressions for the optimal PA orientation (pitch and roll angles) and the optimal user antenna polarization vector. We further reveal the fundamental trade-off among waveguide attenuation, atmospheric absorption, and geometric spreading, leading to an analytical solution for the optimal PA position in single-user scenarios.
    \item Leveraging the inherent spatial selectivity of the radiated modes, we propose a modular optimization framework comprising: (i) geometry-guided user grouping, (ii) Hungarian-based greedy PA-to-group assignment, and (iii) fractional programming (FP) with closed-form polarization updates for joint precoding and receive polarization optimization. This framework scales gracefully to multiple waveguides, PAs, and arbitrary numbers of users.
    \item Comprehensive numerical and simulation results quantify the substantial gains brought by the MMPASS framework: it achieves up to a $167\%$ increase in spectral efficiency compared with single-mode PASS. Moreover, when comparing MMPASS with its polarization-ignorant counterpart, polarization awareness alone improves the sum rate by up to $23\%$. In addition, we show that discrete polarization control recovers most of the continuous optimization gain, offering an attractive complexity-performance trade-off.
\end{itemize}

\section{System Model}\label{sec:II}
We consider a MMPASS, as illustrated in Fig.~\ref{fig:system_model}.
The system comprises $M$ dielectric waveguides, each hosting $N$ PAs. 
Every PA supports $Q$ independent radiating ports, one for each guided mode, enabling simultaneous multi-stream transmission over a single waveguide.
The dielectric waveguides and PAs share an identical rectangular cross section with dimensions $a \times b$. 
Furthermore, we introduce an internal attenuation coefficient $\alpha_{\text{W}}$ and an atmospheric absorption loss coefficient $\alpha_{\text{A}}$ to represent the propagation loss per unit length within the waveguide and in free space, respectively.

\begin{figure}
    \centering
    \includegraphics[width=0.8\linewidth]{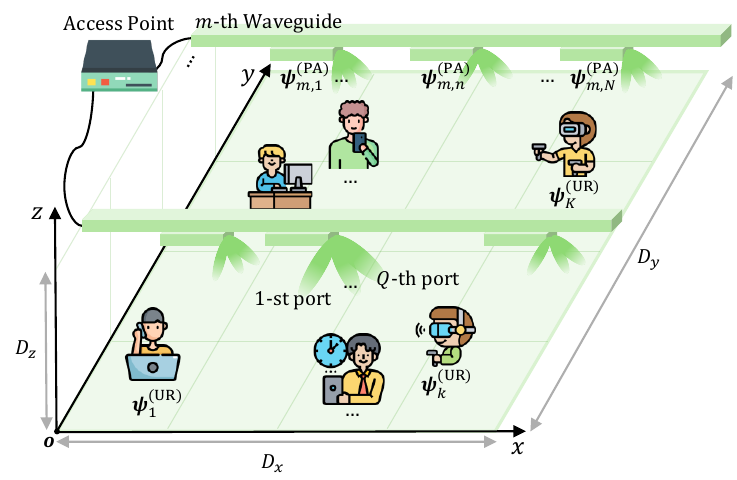}
    \caption{$M$ waveguides are deployed along the $x$-axis, each hosting $N$ PAs. Every PA is equipped with $Q$ independently orientable ports (one per guided mode), enabling simultaneous multi-stream transmission over a single waveguide.}
    \label{fig:system_model}
\end{figure}

\subsection{Spatial Configuration}
The service region is a 3D volume of size $D_x\times D_y\times D_z$, and $K$ single-antenna users are uniformly distributed on the plane $z=0$. A global Cartesian coordinate system (GCS) with basis vectors ${\bm i,\bm j,\bm k}$ is used throughout. Any point in the GCS is represented by the coordinate tuple $\bm{\psi} = (x, y, z)$.

\subsubsection{Waveguides (WG)}
The $M$ waveguides are placed horizontally at height $D_z$, oriented parallel to the $x$-axis and uniformly spaced along the $y$-axis. Each waveguide has a physical length of $D_x$.
The position of the $m$-th waveguide ($m=1,\ldots,M$) is specified by its feed point, denoted by $\bm{\psi}^{(\text{WG})}_m \triangleq \left(x^{(\text{WG})}_m, y^{(\text{WG})}_m, z^{(\text{WG})}_m\right)$.

\subsubsection{PAs}
On waveguide $m$, $N$ PAs are deployed. The radiation center of the $n$-th PA ($n=1,\dots,N$) is denoted by $\bm{\psi}^{ (\mathrm{PA}) }_{m,n} \triangleq \big(x^{(\text{PA})}_{m,n}, y^{(\text{PA})}_{m,n}, z^{(\text{PA})}_{m,n}\big) $, where $y^{(\text{PA})}_{m,n}=y^{(\text{WG})}_m$ and $z^{(\text{PA})}_{m,n}=D_z$.
To realize mode-division multiplexing, each PA is equipped with $Q$ independent ports, each port dedicated to one guided mode. Because the spatial separation among the ports of a PA is negligible compared to the distances to the users, all $Q$ ports share the same radiation center $\bm{\psi}^{ (\mathrm{PA}) }_{m,n}$. However, each port can be individually oriented to steer its radiation pattern. The orientation of the $q$-th port ($q=1,\dots,Q$) is defined by two angles: a pitch angle $\delta^{(\text{PA})}_{m,n,q}$ (rotation about the $y$-axis) and a roll angle $\xi^{(\text{PA})}_{m,n,q}$ (rotation about the $x$-axis). This directional flexibility is essential for both spatial selectivity and polarization control, as will be seen later.

\subsubsection{Users (UR)}
The $K$ users are located on the plane $z=0$; the position of the $k$-th user ($k=1,\dots,K$) is
$ \bm{\psi}^{(\text{UR})}_k \triangleq \big(x^{(\text{UR})}_k, y^{(\text{UR})}_k, 0\big)$.

To characterize the directional radiation of each port, we introduce two auxiliary coordinate systems: a local Cartesian coordinate system (LCS) attached to the port, and a local spherical coordinate system (LSCS) for describing directions.

\begin{defi}[LCS]
\label{defi:LCS_GCS}
For the $q$-th port of the $n$-th PA on the $m$-th waveguide, the LCS is defined with origin at $\bm{\psi}^{(\mathrm{PA})}_{m,n}$ and with axes $\bar{\bm{i}}$, $\bar{\bm{j}}$, and $\bar{\bm{k}}$ such that the port's pointing direction coincides with $-\bar{\bm k}$.
The transformation from GCS coordinates $\bm{\psi}$ to LCS coordinates is
\begin{equation}\label{coord_transform}
\bar{\bm{\psi}}_{m,n,q}
=
\bm{R}_y\!\left(\delta^{(\text{PA})}_{m,n,q}\right)^{\top}
\bm{R}_x\!\left(\xi^{(\text{PA})}_{m,n,q}\right)^{\top}
\big(\bm{\psi}-\bm{\psi}^{(\text{PA})}_{m,n}\big),
\end{equation}
where 
$$\bm{R}_x  (\theta) \!= \!\! \begin{bmatrix}
1\!\!\! & 0\!\!\!  & 0\\
0\!\!\! & \cos \theta\!\!\!  & -\sin\theta\\
0\!\!\! & \sin\theta\!\!\!  & \cos\theta
\end{bmatrix} 
\text{ and }
\bm{R}_y (\theta) \! = \!\!  \begin{bmatrix}
\cos\theta \!\!\! & 0 \!\!\! & \sin\theta\\
0 \!\!\! & 1 \!\!\! & 0\\
-\sin\theta \!\!\! & 0 \!\!\! & \cos\theta
\end{bmatrix}$$
are the rotation matrices about the $x$- and $y$-axes, respectively.
\end{defi}

\begin{defi}[LSCS]
\label{defi:local_spherical}
For a point with LCS coordinates $\bar{\bm{\psi}}_{m,n,q}$, the LSCS uses the radial distance $\bar{r}_{m,n,q} \triangleq \|\bar{\bm{\psi}}_{m,n,q}\|$, polar angle $\theta_{m,n,q} \triangleq \arccos(\bar{z}_{m,n,q}/\bar{r}_{m,n,q})$, and azimuth angle $\phi_{m,n,q} \triangleq \mathrm{atan2}(\bar{y}_{m,n,q}, \bar{x}_{m,n,q})$.\footnote{Here, $\mathrm{atan2}(\bar{y}, \bar{x})$ denotes the four-quadrant inverse tangent, returning the azimuth angle of the point in the $xy$-plane relative to the positive $x$-axis.} 
The transformation between the LCS and LSCS coordinate bases is given by
\begin{equation}
\begin{bmatrix}
\bar{\bm{\upsilon}}\\
\bar{\bm{\vartheta}}\\
\bar{\bm{\varphi}}
\end{bmatrix}
=
\begin{bmatrix}
\sin\theta\cos\phi & \sin\theta\sin\phi & \cos\theta\\
\cos\theta\cos\phi & \cos\theta\sin\phi & -\sin\theta\\
-\sin\phi & \cos\phi & 0
\end{bmatrix}
\begin{bmatrix}
\bar{\bm{i}}\\
\bar{\bm{j}}\\
\bar{\bm{k}}
\end{bmatrix}.
\end{equation}
\end{defi}

\begin{figure}
    \centering
    \includegraphics[width=0.6\linewidth]{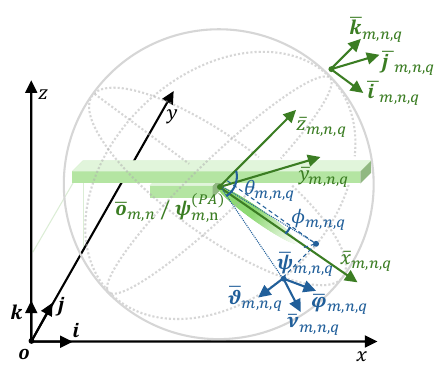}
    \caption{The relationships among GCS, LCS and LSCS, where $\mathbf{o}$, $\bar{\mathbf{o}} \,/\, \bm{\psi}_{m,n}^{(\mathrm{PA})}$, and $\bar{\bm{\psi}}$ are their respective origins.}
    \label{fig:coordinate_system}
\end{figure}

Fig.~\ref{fig:coordinate_system} illustrates the relationships among GCS, LCS and LSCS. With these definitions, the radiated electric field from any port can be expressed compactly in the LSCS, which is natural for far-field propagation and polarization description.

\subsection{Signal Model in MMPASS}
Let $\bm{s}_0 \in \mathbb{C}^{K \times 1}$ be the vector of independent data symbols intended for the $K$ users, with $\mathbb{E}\{\bm{s}_0 \bm{s}_0^{\mathrm{H}}\} = \bm{I}_K$.
A digital precoder $\bm{W} \in \mathbb{C}^{QM \times K}$ maps the symbols to the $QM$ mode-domain inputs:
\begin{equation}\label{e:s}
\bm{s} = \bm{W}\bm{s}_0,\quad \operatorname{tr}(\bm{W}\bm{W}^{\mathsf H})\le 1.
\end{equation}
The entries of $\bm{s}$ are arranged as $s_{m,q}$, the excitation of the $q$-th guided mode on the $m$-th waveguide. The $QM$ mode inputs can be launched via multi-mode couplers (e.g., orthomode transducers) that preserve modal orthogonality.

The received signal at all users, after propagation through the waveguide, radiation into free space, and polarization-dependent reception, can be written as
\begin{equation} \label{eq:end_to_end}
\bm{\zeta}
=
\sqrt{P}\bm{H}\bm{s} +\bm{\nu} 
\triangleq \sqrt{P}  \left(\bm{\Lambda}\odot \bm{H}^{(\mathrm{P}\to \mathrm{U})}\right) \bm{H} ^ {(\mathrm{W}\to \mathrm{P})} \bm{s}
+\bm{\nu},
\end{equation}
where $P$ is the total transmit power, $\bm{\nu}\sim\mathcal{CN}(0,\sigma^2\bm{I}_K)$ is the additive white Gaussian noise (AWGN), $\bm{H}\in\mathbb{C}^{K\times QM}$ is the effective end-to-end channel matrix, and $\odot$ denotes the Hadamard product.

To expose the unique propagation characteristics of MMPASS, \eqref{eq:end_to_end} decomposes $\bm H$ into three components.
\begin{itemize}[leftmargin = 0.4cm]
    \item \textit{Waveguide attenuation:} $\bm{H}^{(\mathrm{W}\to \mathrm{P})} \in \mathbb{C}^{QMN \times QM}$ captures guided-wave propagation from the $QM$ mode inputs to the $QMN$ radiating ports (one per port per PA). Unlike prior studies that often omit waveguide losses \cite{11358830,11371592}, we explicitly incorporate the exponential attenuation. 
    \item \textit{Electromagnetic radiation:} $\bm{H}^{(\mathrm{P}\to \mathrm{U})} \in \mathbb{C}^{K \times QMN}$ models the free-space radiation from each port to each user, including the directional pattern and phase progression. This is in contrast to prevailing prior works that use a simplified omnidirectional radiation model \cite{wang2025modeling,11223640}.
    \item \textit{Polarization matching:} $\bm{\Lambda} \in \mathbb{R}^{K \times QMN}$ represents polarization matching between the incident electric field at each user and the user's receive antenna.
\end{itemize}

\section{Full-Wave Modeling of the MMPASS Channel}\label{sec:III}
In this section, we present rigorous electromagnetic formulations of $\mathbf{H}^{(\mathrm{W}\to \mathrm{P})}$, $\mathbf{H}^{(\mathrm{P}\to \mathrm{U})}$, and $\bm{\Lambda}$ in \eqref{eq:end_to_end}.
Unlike existing black-box CMT approaches that characterize only power transfer, our model derives the actual electric field distributions inside the multi-mode waveguide, the radiated fields from each PA port, and the polarization-dependent reception at the user side. This full-wave formulation provides the necessary physical foundation for polarization-aware optimization and reveals the spatial selectivity inherent in multi-mode PASS.

We focus on transverse electric (TE) modes, which offer well-defined polarization states and controllable radiation characteristics. In a TE mode, the electric field is strictly perpendicular to the propagation direction. Let $\text{TE}_{u,v}$ denote the supported TE modes, where $u$ and $v$ indicate the field variations along the waveguide's transverse $y$- and $z$-dimensions, respectively. All admissible $(u,v)$ pairs are collected into a set of $Q$ modes and indexed by a single mode index $q$, establishing a one-to-one correspondence between $(u,v)$ and $q$. Fig.~\ref{fig:TE_mode} illustrates the $\text{TE}_{1,0}$ mode, which serves as the dominant mode in rectangular dielectric waveguides.

\begin{figure}
    \centering
    \includegraphics[width=0.8\linewidth]{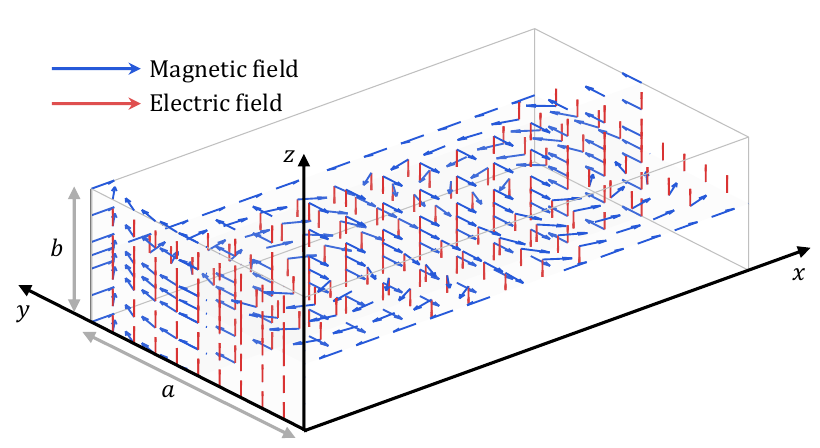}
    \caption{Electric field distribution of the $\text{TE}_{1,0}$ mode in a rectangular waveguide.}
    \label{fig:TE_mode}
\end{figure}

\subsection{Waveguide Propagation}
Unlike existing PASS literature that treats the waveguide as a lossless transmission line with a single scalar coefficient, we explicitly model the spatial field distribution of each guided mode. 

For the $m$-th waveguide, the transverse electric field of the $q$-th TE mode at an arbitrary point $\bm{\psi}=(x,y,z)$ inside the waveguide is given by~\cite{pozar2021microwave}
\begin{eqnarray} \label{eq:Et_mode1}
&&\hspace{-0.8cm} \bm{E}_{m,q}^{(\text{WG})}(\bm{\psi})  =  \frac{j \omega \mu \pi s_{m,q}}{\varrho_{q}^2 }  \sqrt{e^{ - \alpha_{\text {W}} x}}  e^{-j \beta_q x } \notag \\
 &&\hspace{-0.8cm} \cdot \! \Bigg[ \!  \frac{v}{b} \! \cos \! \left( \!  \frac{u\pi}{a} \! \left( \!\! \left(y \!- \! y^{(\text{WG})}_m\right) \! + \! \frac{a}{2} \! \right) \!\! \right) \! \sin \! \left( \!  \frac{v\pi}{b} \! \left( \!\!  \left( \! z \! - \! z^{(\text{WG})}_m \! \right) \! + \! \frac{b}{2} \!  \right) \!\! \right) \! {\bm{j}} \notag\\
 &&\hspace{-0.8cm} +  \frac{u}{a} \!\sin \! \left( \! \frac{u\pi}{a} \! \left( \!\! \left( \! y \! - \!  y^{(\text{WG})}_m \! \right)\!+ \!\frac{a}{2} \! \right) \!\! \right) \! \cos \! \left( \! \frac{v\pi}{b} \! \left( \!\! \left( \! z \! - \! z^{(\text{WG})}_m \! \right) \! +\frac{b}{2}\! \right) \!\! \right) \! {\bm{k}} \! \Bigg] \! \notag \\
 &&\hspace{-0.8cm} \triangleq \bm{\mathcal{E}}_{m,q}^{(\text{WG})}(\bm{\psi}) \sqrt{e^{ - \alpha_{\text{W}} x}}  e^{-j \beta_q x },
\end{eqnarray}
where $\omega = 2\pi f$ is the angular frequency,
$\mu$ is the permeability of the dielectric waveguide,
$\varrho_q = \sqrt{(u\pi/a)^2 + (v\pi/b)^2}$ is the cutoff wavenumber,
$\varrho = \frac{2\pi}{\lambda}$ is the wavenumber,
$\beta_q = \sqrt{\varrho^2 - \varrho_q^2}$ is the propagation constant,
and $\lambda$ is the wavelength.

The factor $\sqrt{e^{-\alpha_{\text{W}} x}}$ captures the internal waveguide attenuation, while $e^{-j \beta_q x}$ accounts for the phase progression. The spatial pattern $\bm{\mathcal{E}}_{m,q}^{(\text{WG})}(\bm{\psi})$ contains the transverse field distribution of mode $q$.

Owing to mode orthogonality, the total field inside the $m$-th waveguide is the superposition of all $Q$ modes: $\bm{E}_m^{(\text{WG)}}(\bm{\psi})
= \sum_{q=1}^Q {\bm{E}_{m,q}^{(\text{WG})}(\bm{\psi}) }$.

For each PA, the physical interval along the waveguide over which it extracts guided energy is referred to as the coupling length \cite{wang2025modeling}. Specifically, for the $n$-th PA on the $m$-th waveguide, this length is denoted by $\tau_{m,n}$, and the fraction of power coupled out from the waveguide is given by $\sin^2(\kappa \tau_{m,n})$, where $\kappa$ is the coupling coefficient determined by the waveguide material properties. The mode-separated electric field at the radiation point of the PA can then be expressed as
\begin{eqnarray} \label{eq:E_PA^t}
 \bm{E}_{m,n,q}^{(\text{PA})} &\hspace{-0.4cm} = &\hspace{-0.4cm} \left( \prod_{i=1}^{n-1}\!\sqrt{1\!-\!\sin^{\!2}\!\left(\kappa\tau_{m,i}\right)}\!\right)  \! \sin \! \left(\kappa\tau_{m,n}\right) \!  \bm{E}_{m,q}^{(\text{WG)}} \! \left(\bm{\psi}^{(\text{PA})}_{m,n} \right) \notag \\
 &\hspace{-0.4cm} \overset{(a)}{=} &\hspace{-0.3cm} \sqrt{\frac{1}{N}}\cdot\bm{E}_{m,q}^{(\text{WG})}\left(\bm{\psi}^{(\text{PA})}_{m,n} \right).
\end{eqnarray}
The first term represents the residual field after coupling by PAs $1, \dots, n-1$, while the second term corresponds to the fraction coupled out by the $n$-th PA. Step (a) follows from employing an equal-quota power allocation strategy in~\cite{zhang2025directional}, under which the coupling length is set as $\tau_{m,n} = \frac{1}{\kappa}\arcsin\sqrt{\frac{1}{N+1-n}}$.

Combining \eqref{eq:Et_mode1} and \eqref{eq:E_PA^t}, the channel gain from the waveguide input to the PA radiation point follows directly.
\begin{thm}[Waveguide-to-PA Channel Gain in MMPASS]\label{Thm:h_WP}
    The channel coefficient for mode $q$ from the input of the $m$-th waveguide to the radiation point of the $n$-th PA on the same waveguide is given by
    \begin{equation} \label{eq:h^WP}
    \hspace{-0.3cm}h^{(\text{W }\to\text{ P})}_{m,n,q}\!  = \! \frac{\bm{E}_{m,n,q}^{(\text{PA})}}{\bm{E}_{m,q}^{(\text{WG})} \!\! \left(\bm{\psi}^ {(\text{WG})}_m  \right)}
    \! = \! \sqrt{\frac{e^{ - \alpha_{\text{W}} x^{(\text{PA})}_{m,n}}}{N}} e^{-  j\beta_q x^{(\text{PA})}_{m,n}}\!.\!\!
    \end{equation}
\end{thm}

From Theorem \ref{Thm:h_WP}, the channel matrix $\bm{H}^{(\text{W}\to\text{ P})}$ is block diagonal with $M$ blocks of size $NQ \times Q$, since each PA couples energy exclusively from its hosting waveguide. Each block further comprises $N$ diagonal submatrices of size $Q \times Q$, reflecting the orthogonality among the supported modes.

\subsection{PA Radiation}

We now derive the electric field radiated from the PA to an arbitrary observation point. The derivation follows electromagnetic radiation theory and proceeds as follows:
\begin{itemize}[leftmargin= 0.4cm]
    \item First, given the aperture field distribution in \eqref{eq:E_PA^t}, the equivalent surface electric and magnetic current densities, $\bm{J}_s$ and $\bm{M}_s$, can be obtained via the Huygens-Love equivalence principle \cite{schwinger2019classical}.
    \item The radiated fields are then determined using the free-space Green's function \cite{schwinger2019classical}. The magnetic vector potential $\bm{A}$ due to $\bm{J}_s$ and, by electromagnetic duality, the electric vector potential $\bm{F}$ due to $\bm{M}_s$ can then be expressed in integral form over the aperture.
    \item Finally, the electric field components are obtained from the corresponding differential relations applied to $\bm{A}$ and $\bm{F}$. Denoting the resulting fields by $\bm{E}_A$ and $\bm{E}_F$, respectively, the total field at any observation point follows from linear superposition as $\bm{E} = \bm{E}_A + \bm{E}_F$.
\end{itemize}

Under the far-field condition, i.e.$\|\bm{\psi}\|\gg\frac{D_s^{2}}{\lambda}$, where $D_s$ denotes the dimension of the radiation plane, the radiated field expression can be further simplified \cite{stutzman2012antenna}. Thus, the electric field induced at an arbitrary observation point $\bm{\psi}$ by the $q$-th port of the $n$-th PA on the $m$-th waveguide can be characterized in a unified form.

\begin{thm}[PA Radiation Model]\label{thm:E_radiation}
Consider the $q$-th port of the $n$-th PA on the $m$-th waveguide. Let $\bar{\bm{\psi}}_{m,n,q}$ be the relative position vector of an observation point $\bm{\psi}$ in the port's LCS, with local polar angle $\theta_{m,n,q}$ and azimuth $\phi_{m,n,q}$. The electric field radiated by mode $q$ at $\bm{\psi}$ is 
\begin{eqnarray}\label{eq:E_field_rigorous}
    &&\hspace{-1cm}\bm{E}_{m,n,q}(\bm{\psi}) =  \frac{j \varrho a b \omega \mu s_{m,q}  }{2\varrho_{q}^2 \pi \sqrt{N} \|\bar{\bm{\psi}}_{m,n,q}\| }   e^{\!  - \frac{1}{2}  \left(  \alpha_{\text{W}} x^{(\text{PA})}_{m,n} +  \alpha_{\text{A}} \|\bar{\bm{\psi}}_{m,n,q}\|   \right)  }\\
    &&\hspace{1cm}    e^{-j \left(\beta_q x^{(\text{PA})}_{m,n}+\varrho\|\bar{\bm{\psi}}_{m,n,q}\|\right)} \mathcal{S}_q (\theta_{m,n,q}, \phi_{m,n,q})  \bm{\Psi}_q ,\notag
\end{eqnarray}
where $\mathcal{S}_q$ denotes the pattern factor, which characterizes the spatial distribution and diffraction behavior of the rectangular PA aperture;
$\bm{\Psi}_q$ is the polarization vector, capturing the transverse polarization state of the radiated electric field within the LSCS.
\end{thm}

\begin{proof}
(sketch) The far-field radiation at an arbitrary observation point is evaluated by superposing the equivalent electric and magnetic surface currents on the aperture. Driven by the tangential aperture fields, these induced current densities are expressed as follow. For notational simplicity, we omit the subscripts ${m,n,q}$ whenever no ambiguity arises.
\begin{equation}
  \bm{J}_{s} =\frac{1}{\sqrt{N}} \bm{n} \times \bm{\mathcal{H}}^{(\text{WG})}(\bm p),  \label{eq:J_s}
\end{equation}
\begin{equation}
    \bm{M}_s = -\frac{1}{\sqrt{N}}  \bm{n}\times \bm{E}^{(\text{WG})}(\bm p) ,\label{eq:M_s}
\end{equation}
where  $\bm{n}$ denotes the rectangular aperture $S$ of the $q$-th port outward unit normal in the LCS, and $\bm p \in \mathcal S$ denotes the position vector of an arbitrary point on the aperture plane.
$\bm{\mathcal{H}}^{(\text{WG})}(\bm p) \!\!=\!\! \frac{\beta}{\omega\mu} \big({\bm{n}} \!\times\!  \bm{E}^{(\text{WG})}(\bm p) \big)$ denotes the tangential magnetic field associated with the tangential electric field $ \bm{E}^{(\text{WG})}(\bm p)$.

In free space, the equivalent electric and magnetic surface currents generate the electric and magnetic vector potentials
\begin{equation*}
\bm{A}(\bm \psi)=\mu\,\frac{e^{-j\varrho\| \bar{\bm \psi} \|}}{4\pi \| \bar{\bm \psi} \|}
\iint_{S} \bm{J}_s(\bm p)\,e^{j\varrho\,\frac{\bar{\bm \psi}}{\| \bar{\bm \psi} \|}\cdot \bm p}\,\text{d}S, \label{e:A} \\
\end{equation*}
\begin{equation*}
    \bm{F}(\bm \psi)=\varepsilon\,\frac{e^{-j\varrho\| \bar{\bm \psi} \|}}{4\pi \| \bar{\bm \psi} \|}
\iint_{S} \bm{M}_s(\bm p)\,e^{j\varrho\,\frac{\bar{\bm \psi}}{\| \bar{\bm \psi} \|}\cdot \bm p}\,\text{d}S,\label{e:F}
\end{equation*}
where $\varepsilon$ denotes the permittivity of the surrounding propagation medium. Under the Lorenz gauge condition and Helmholtz equations, the electric field can be written in terms of the vector potentials as
\begin{eqnarray} \label{eq:Epsi}
     &&\hspace{-0.8cm}\bm{E}({\bm{\psi}}) =  \frac{1}{j\omega\varepsilon}\nabla\times\nabla\times\bm A(\bm{\psi}) - \frac{1}{\varepsilon}\nabla\times\bm F(\bm{\psi}) \\
    &&\hspace{-1cm}= -j\omega \mu  \frac{\bar{\bm \psi}}{\| \bar{\bm \psi} \|} \times  \left( \frac{\bar{\bm \psi}}{\| \bar{\bm \psi} \|} \times \bm{A}(\bm \psi) +   \sqrt{\frac{\mu}{\varepsilon}}  \bm{F}(\bm \psi) \right).\notag
\end{eqnarray}
Then, \eqref{eq:E_field_rigorous} can be obtained after some manipulations. 
A detailed derivation is provided in Appendix \ref{sec:AppA}.
\end{proof}

For the practically important case of two orthogonal modes $\text{TE}_{1,0}$ and $\text{TE}_{0,1}$, the pattern factors and polarization vectors admit closed forms.

\begin{cor}
    Omitting the waveguide index $m$ and the PA index $n$ for brevity, the spatial pattern factors $\mathcal{S}_q$ for a two-mode configuration excited by the orthogonal $\text{TE}_{1,0}$ and $\text{TE}_{0,1}$ modes ($q = 1, 2$) are given by
    \begin{equation}
    \hspace{-0.1cm}\begin{cases}
    \mathcal{S}_1(\theta_{1},\phi_{1})
    =
    \frac{\sin(\frac{b\pi}{\lambda}\sin\theta_{1}\cos\phi_{1})}{\frac{b\pi}{\lambda}\sin\theta_{1}\cos\phi_{1}}  \cdot 
    \frac{\cos(\frac{a\pi}{\lambda}\sin\theta_{1}\sin\phi_{1})}{1-(\frac{2a}{\lambda}{\sin\theta_{1}\sin\phi_{1}})^2}, \\
    \mathcal{S}_2(\theta_{2},\phi_{2})
    =
    \frac{\sin(\frac{a\pi}{\lambda}\sin\theta_{2}\sin\phi_{2})}{\frac{a\pi}{\lambda}\sin\theta_{2}\sin\phi_{2}}  \cdot 
    \frac{\cos(\frac{b\pi}{\lambda}\sin\theta_{2}\cos\phi_{2})}{1-(\frac{2b}{\lambda}\sin\theta_{2}\cos\phi_{2})^2}.
    \end{cases} \!\!\!
    \end{equation}
While the associated two-mode transverse polarization is
\begin{equation}\label{eq:obliquity_vector}
\hspace{-0.2cm}\begin{cases}
\! \bm{\Psi}_1 \! =\! \left(\!1 \!+ \!\frac{\beta_1}{\varrho}\cos\theta_{1}\!\right)
\!\cos\!\phi_{1}\bar{\bm{\vartheta}}_1 \!+\! \left(\frac{\beta_1}{\varrho} + \cos\theta_{1}\right)
\sin\!\phi_{1}\bar{\bm{\varphi}}_1, \\
\! \bm{\Psi}_2 \! =\! \left(\!1 \!+ \!\frac{\beta_2}{\varrho}\cos\theta_{2}\!\right)
\!\sin\!\phi_{2}\bar{\bm{\vartheta}}_2 \!+\! \left(\frac{\beta_2}{\varrho} + \cos\theta_{2}\right)
\cos\!\phi_{2}\bar{\bm{\varphi}}_2,
\end{cases} \!\!\!\!\!\!\!\!\!
\end{equation}
where $\bar{\bm{\vartheta}}$ and $\bar{\bm{\varphi}}$ are the local spherical unit vectors spanning the transverse plane perpendicular to the propagation direction.
\end{cor}

\begin{figure}[!tb]
    \centering
    \includegraphics[width=0.8\linewidth]{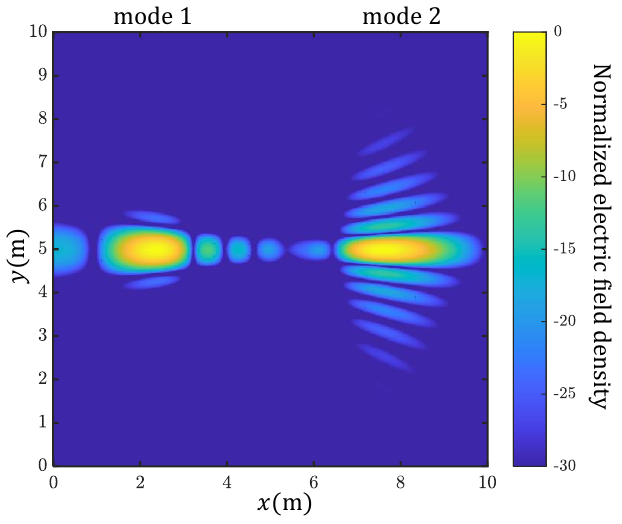}
    \caption{Spatial distribution of the normalized radiated electric-field intensity (dB scale) for a dual-mode PA. The two ports generate spatially separated beams.}
    \label{fig:energy_map}
\end{figure}

Fig. \ref{fig:energy_map} illustrates the normalized radiated intensity, $\|\bm E_{m,n,q}(\bm{\psi})\|^2$ on the $z=0$ plane for a dual-mode PA positioned at the ceiling center, with port angles set to $\pm \frac{\pi}{4}$. 
The two ports produce distinct, spatially separated beams, confirming the inherent spatial selectivity of MMPASS. This selectivity is key to reducing inter-user interference in multi-user scenarios.

Leveraging \eqref{eq:E_PA^t} and \eqref{eq:E_field_rigorous}, we can obtain the electric fields of the $q$-th mode at the PA $\bm{\psi}^{(\text{PA})}_{m,n}$ and the user $\bm{\psi}_{k}^{(\text{UR})}$, respectively, which further enables the derivation of the wireless channel gain. Given that \eqref{eq:E_PA^t} is formulated in the Cartesian coordinate system while \eqref{eq:E_field_rigorous} employs spherical coordinates, we characterize the resulting channel gain in Theorem \ref{prop:H_PU} by independently accounting for the amplitude attenuation and the phase shift encountered across the wireless propagation path.

\begin{thm}[PA-to-User Channel Gain in MMPASS]\label{prop:H_PU}
Considering transmission from the $q$-th port of the $(m,n)$-th PA to the $k$-th user, the channel gain can be expressed as
\begin{eqnarray} \label{eq:h_channel}
&&\hspace{-0.8cm}h^{(\text{P }\to\text{ U})}_{m,n,q,k}  =  \frac{\left\|\bm{E}_{m,n,q}(\bm{\psi}_{k}^{(\text{UR})}) \right\|}{\left\|\bm{E}_{m,n,q}^{(\text{PA})}\right\|} e^{-j\varrho\|\bar{\bm{\psi}}_{m,n,k}^{(\text{UR})} \|} \\
&&\hspace{-0.8cm} = \! \frac{  a b \omega \mu  s_{m,q}  }{\lambda \varrho^2_q   \|\bar{\bm{\psi}}_{m,n,q}\| } \frac{ \mathcal{S}_q (\theta_{m,n,q}, \phi_{m,n,q})  \left\| \bm{\Psi}_q \right\|}{\left\|\bm{\mathcal{E}}_{m,q}^{(\text{WG})}(\bm{\psi}^{(\text{PA})}_{m,n} ) \right\|} e^{-\|\bar{\bm{\psi}}_{m,n,k}^{(\text{UR})} \|(j\varrho+\frac{1}{2}\alpha_{\text{A}})}, \notag
\end{eqnarray}
where $\bar{ \bm{\psi} }_{m,n,k}^{ (\text{UR}) }$ is the coordinates of the receiving antenna within the LCS centered at $\bm{\psi}^{(\text{PA})}_{m,n}$.
\end{thm}

By aggregating the individual channel coefficients from \eqref{eq:h_channel} across all users and PA ports, we obtain the global wireless channel matrix $\bm{H}^{(\mathrm{P} \to \mathrm{U})} \in \mathbb{C}^{K \times QMN}$ characterizing the propagation links between every radiating port across the $M$ waveguides and all $K$ users. Specifically, $h^{(\mathrm{P} \to \mathrm{U})}_{m,n,q,k}$ defined in \eqref{eq:h_channel} is mapped to the entry at the $k$-th row and the $[(m-1)NQ + (n-1)Q + q]$-th column.

\subsection{Polarization-Aware Reception}\label{subsec:effective_E_rx}
While the previous subsection established the radiation patterns and intrinsic polarization states, this subsection incorporates the dynamic effects of port rotation. By utilizing Jones matrices, we characterize the polarization-aware signal reception, thereby capturing the polarization matching efficiency as a function of the relative orientation between the PA ports and the receiving antenna.

Given the radiated field in \eqref{eq:E_field_rigorous}, we define two specific components in LSCS to yield a more intuitive formulation as 
\begin{eqnarray}\label{e:E_LSCS}
    \hspace{-0.3cm}\bm{E}_{m,n,q}(\bm{\psi}) &\hspace{-0.3cm}\triangleq&\hspace{-0.3cm} E_{m,n,q}^{\theta}(\bm{\psi})\bar{\bm{\vartheta}}_{m,n,q} + E_{m,n,q}^{\phi}(\bm{\psi})\bar{\bm{\varphi}}_{m,n,q}\\
    \hspace{-0.3cm}&\hspace{-0.3cm}\overset{(a)}{=}&\hspace{-0.3cm} E_{m,n,q}^{\theta}(\bm{\psi})\bar{\bm{\vartheta}}_{m,n,q,k} - E_{m,n,q}^{\phi}(\bm{\psi})\bar{\bm{\varphi}}_{m,n,q,k}.\notag
\end{eqnarray}
where (a) follows from adopting the LSCS centered at $\bm{\psi}_{k}^{(\text{UR})}$ to characterize the incident polarization state from the perspective of the receiving antenna. Then, the normalized Jones vector of the incident electric field can be obtained as
\begin{equation}\label{e:incident_Jones}
    \bar{\bm{n}}_{m,n,q}^{(i)}(\bm{\psi}) = \begin{bmatrix}
    \frac{E_{m,n,q}^{\theta}(\bm{\psi})}{\left\|\bm{E}_{m,n,q}(\bm{\psi}) \right\| }
    & \frac{- E_{m,n,q}^{\phi}(\bm{\psi})}{\left\|\bm{E}_{m,n,q}(\bm{\psi}) \right\| }
\end{bmatrix}\!\!\!
\begin{bmatrix}
\bar{\bm{\vartheta}}_{m,n,q,k}\\\bar{\bm{\varphi}}_{m,n,q,k}
\end{bmatrix}.
\end{equation}
which is a linear combination of the unit basis vectors corresponding to the polar and azimuthal angles. Therefore, it inherently resides in the transverse plane orthogonal to the propagation vector.

The polarization matching efficiency, formulated in Theorem \ref{lem:polarization_efficiency}, is derived by projecting the incident Jones vector from \eqref{e:incident_Jones} onto the antenna's normalized Jones vector $\bar{\bm{n}}_{m,n,k}$. Under the assumption of linear polarization for both the radiated field and the user antenna, which is standard for most radiating elements such as dipoles, this efficiency reduces to the squared magnitude of the inner product between the two vectors \cite{goldstein2017polarized}.

\begin{thm}[Polarization Matching Coefficient]\label{lem:polarization_efficiency}
Consider an incident wave originating from the $q$-th mode of the $n$-th port on the $m$-th waveguide. Let $\bar{\bm{n}}_{m,n,q}^{(i)}(\bm{\psi}_{k}^{(\mathrm{UR})})$ and $\bar{\bm{n}}_{m,n,k}^{(r)}$ denote the normalized Jones vectors of the incident field and the antenna of user $k$, respectively, both defined within the $k$-th user's LSCS. The resulting polarization matching efficiency, which quantifies the fraction of power captured by the antenna, is given by
\begin{equation} \label{eq:plf_def}
    \eta_{m,n,q,k} = \big|(\bar{\bm{n}}_{m,n,q,k}^{(r)})^\top \bar{\bm{n}}_{m,n,q}^{(i)}(\bm{\psi}_{k}^{(\mathrm{UR})}) \big| .
\end{equation}
\end{thm}

Finally, incorporating the analysis from the preceding subsections, the received signal for $s_{m,q}$ propagating through the $n$-th PA to the $k$-th user is expressed as:
\begin{equation}\label{eq:r_dual_compact}
r_{m,n,q,k}
=
\eta_{m,n,q,k}\cdot h^{(\text{P} \to\text{U})}_{m,n,q,k} \cdot h^{(\text {W} \to\text{P})}_{m,n,q}\cdot s_{m,q}.
\end{equation}
Because distinct guided modes have different polarization states and radiation patterns, the polarization matching efficiency varies with the mode index $q$ and the user's location. Moreover, all $Q$ ports of a given PA share the same physical position, coupling the optimization across modes. These unique characteristics distinguish MMPASS from conventional multi-antenna systems and motivate the optimization framework developed in Sec.~\ref{sec:IV}.

\section{Sum-Rate Maximization with MMPASS}\label{sec:IV}
The physically complete electromagnetic model developed in the previous section reveals that MMPASS is determined by a rich set of design variables. Exploiting the high-dimensional design space is essential for achieving the full potential of MMPASS. In this section, we formulate the sum-rate maximization problem and develop a structured optimization framework that proceeds from fundamental single-user insights to general multi-user, multi-PA configurations.

The optimization problem is formulated as follows:
\begin{subequations}\label{prob:sumrate}
\begin{align}
\hspace{-0.5cm}\text{(P1):} \quad & \!\!\!\!\! \max_{ \{\!x_{m,n}^{(\text{PA})}\!\}, \{\!\xi_{m,n,q}^{(\text{PA})}\!\},\{\!\delta_{m,n,q}^{(\text{PA})}\!\},\bm{W},\,\{\bm{n}_k^{(r)}\}} \quad \!\!\!\!\!\!\!\!\!\!\!\! R_{\mathrm{sum}} =
\sum_{k=1}^{K} R_k, \label{e:P1a}
 \\
\hspace{-0.5cm}\text{s.t.} \quad  \hspace{0.2cm}& \hspace{1cm}\mathrm{tr}\!\left(\bm{W}\bm{W}^{\mathsf H}\right)\le 1, \label{e:P1b}\\
\hspace{-0.5cm}&\hspace{1cm} 0<x_{m,n}^{(\text{PA})} \leq D_x, \label{e:P1c}\\
\hspace{-0.5cm}&\hspace{1cm} x_{m,(n+1)}^{(\text{PA})}- x_{m,n}^{(\text{PA})} \geq \frac{\lambda}{2}, \label{e:P1d}\\
\hspace{-0.5cm}&\hspace{1cm} 0 \leq \delta_{m,n,q}^{(\text{PA})} \leq \frac{\pi}{2}, \label{e:P1e}\\
\hspace{-0.5cm}&\hspace{1cm} -\frac{\pi}{2} \leq \xi_{m,n,q}^{(\text{PA})} \leq \frac{\pi}{2}, \label{e:P1f}\\
\hspace{-0.5cm}&\hspace{1cm} \| \bm{n}_k\| = 1, \label{e:P1g}
\end{align}
\end{subequations}
$R_k$ in \eqref{e:P1a} is the achievable data rate of the $k$-th user, which is defined as
\begin{equation}\label{eq:k_rate}
R_k
= \frac{1}{2}\log_2 \left(1+\frac{P\,\big|\bm{h}_k\bm{w}_k\big|^2}
{P\sum_{i\neq k}\big|\bm{h}_k\bm{w}_i\big|^2+\sigma^2}\right).
\end{equation}

Problem (P1) formulates an optimization framework that jointly integrates physical layer design variables and the digital beamforming matrix. The associated constraints ensure the practical implementability and physical consistency of the solution:
\eqref{e:P1c} and \eqref{e:P1d} constrain antenna placement within the waveguide's physical boundaries while imposing a minimum inter-element spacing to mitigate mutual coupling;
\eqref{e:P1e} and \eqref{e:P1f} define the feasible ranges for the rotation angles;
\eqref{e:P1b} and \eqref{e:P1g} enforce normalization requirements on the digital precoders and the receiving polarization vectors, respectively.

To derive analytical insights through closed-form expressions, this section focuses on the practically important case of $Q=2$ modes, specifically the orthogonal $\text{TE}_{1,0}$ and $\text{TE}_{0,1}$ modes. This choice captures the essential features of mode-division multiplexing. 
The extension to higher-order modes follows the same principles.

\subsection{Single-PA Single-User Analysis}\label{sec:IVA}

To gain fundamental insights, we first consider the simplest configuration: $M=N=1$ (one waveguide, one PA) and $K=1$ (one user).
In this case, (P1) reduces to:
\begin{subequations}
\begin{align}
\hspace{-2cm}\text{(P1a):}\hspace{3.5cm}& \hspace{-2.5cm}
\max_{x^{(\text{PA})}, \delta^{(\text{PA})}, \xi^{(\text{PA})}, \bm{n}^{(r)}} 
R, \label{e:2a}\\
\hspace{-2cm}\text{s.t.}\hspace{3.7cm}
& \hspace{-2.3cm} \eqref{e:P1c}, \eqref{e:P1e}, \eqref{e:P1f}, \eqref{e:P1g}.
\label{e:P2b}
\end{align}
\end{subequations}
In this interference-free setup, the sum-rate maximization reduces to the maximization of the received Signal-to-Noise Ratio (SNR). For a given noise power variance $\sigma^2$, this objective is equivalent to maximizing $|H|^2$. Exploiting the monotonicity of the logarithmic function, we can transform the multiplicative components of $H$ into a more tractable additive form. Substituting the explicit expression of $H$ derived in Section \ref{sec:III}, the equivalent optimization objective is given by
\begin{eqnarray}\label{e:ln_H_2}
    &&\hspace{-0.8cm} \ln |H|^2 = 2\ln\left( \frac{\omega \mu a b  }{ \lambda  \varrho_q^2\left\|\bm{\mathcal{E}}_{m,q}^{(\text{WG})}(\bm{\psi}^{(\text{PA})}_{m,n} ) \right\|} \left|\left(\bm {n^{(r)}}\right)^\top\bm {\bm n^{(i)} }\right| \right)\\ 
    &&\hspace{-0.7cm}   + 2\ln\bigg[\frac{\cos(\sin\theta \sin\phi\cdot\frac{a}{\lambda}\pi)}{1-4(\sin\theta \sin\phi \cdot \frac{a}{\lambda})^2} \bigg]  + 2\ln\!\bigg[\frac{\sin(\sin\theta \cos\phi\cdot\frac{b}{\lambda}\pi)}{\sin\theta \cos\phi \cdot \frac{b}{\lambda}\pi} \bigg]\!\notag \\
    &&\hspace{-0.3cm}+ 2\ln(1+\cos\theta)- \alpha_{\text{W}}x^{(\text{PA})}-\alpha_{\text{A}}\| \bar{\bm{\psi}}^{\text{(UR)}}\|- 2 \ln {\| \bar{\bm{\psi}}^{\text{(UR)}}\|}. \notag 
\end{eqnarray}

While all optimization variables in (P1a) are inherently coupled, for any given PA position, the optimal orientation should always maximize the energy delivered to the user. Consequently, the optimal deployment angles can be determined a priori for a fixed $x^{(\text{PA})}$.

\begin{lem}[Optimal PA Orientation for Single User]\label{thm:single_angle}
For a fixed PA position, the orientation angles that maximize the received power are those that align the main lobe of the radiation pattern with the user direction. In the LCS of the PA, let the user's coordinates be $(\bar{x}^{(\text{UR})},\bar{y}^{(\text{UR})},\bar{z}^{(\text{UR})})$. Then the optimal pitch and roll angles are
\begin{align}
    \delta^{(PA)^*} &= \arctan\left( \frac{ \bar{x}^{(\text{UR})} }{ \sqrt{ \left( \bar{y}^{(\text{UR})}  \right)^2 + \left( \bar{z}^{(\text{UR})}   \right)^2 } } \right), \label{e:theta}\\
\xi^{(PA)^*} &= \arctan\left( -\frac{ \bar{y}^{(\text{UR})} }{ \bar{z}^{(\text{UR})}  } \right).\label{e:xi}
\end{align}
\end{lem}

\begin{proof}(sketch)
To maximize the power gain, the following first-order optimality conditions must be satisfied:
\begin{equation}
    \frac{\partial}{\partial \delta^{(\text{PA})}} \ln |H|^2 = 0,
    \quad \frac{\partial}{\partial \xi^{(\text{PA})}} \ln |H|^2 = 0.
\end{equation}
Assuming the radiated energy is primarily concentrated within the main lobe, we approximate these transcendental equations using a Taylor expansion near the main lobe. This leads to the optimal alignment conditions, namely $\bar{x}^{(\text{UR})} = 0$ and $\bar{y}^{(\text{UR})} = 0$, from which the closed-form expressions for $\delta^{(PA)^*}$ and $\xi^{(PA)^*}$ can be derived. The detailed derivation can be found in Appendix~\ref{sec:AppB}.
\end{proof}

Following the angular optimization in Lemma \ref{thm:single_angle}, we proceed to optimize the PA's longitudinal coordinate $(x^{(\text{PA})})^*$. The optimization begins with an analysis of the partial derivative of the channel gain with respect to $x^{(\text{PA})}$:
\begin{eqnarray}\label{e:partial_ln_H_2}
&&\hspace{-1cm}\frac{\partial |H|^2}{\partial x^{(\text{PA})}}
{=} |H|^2\bigg[- \alpha_{\text{W}} + \alpha_{\text{A}}\frac{ \bar{x}^{(\text{UR})}\! }{ \left\| \bar{\bm{\psi}}^{(\text{UR})} \right\| } + 2   \frac{ \bar{x}^{(\text{UR})}  }{ \left\| \bar{\bm{\psi}}^{(\text{UR})} \right\|^2 }\bigg].
\end{eqnarray}
Eq. \eqref{e:partial_ln_H_2} illustrates the competing physical mechanisms governing optimal PA placement, including: (i) {\it waveguide attenuation} imposes a constant penalty per unit length due to propagation losses within the structure; (ii) {\it atmospheric absorption loss} is positively correlated with the wireless propagation distance; and (iii) the increasing wireless distance leads to a {\it larger wavefront}, thereby heightening the free-space path loss.

Collectively, the optimal position is governed by the interplay between the attenuation within the waveguide and the reduction in both atmospheric absorption and geometric spreading losses characteristic of the wireless link. By imposing the optimality condition $\frac{\partial \ln |H|^2}{\partial x^{(\text{PA})}} = 0$, we derive the optimal placement for the single-user scenario.

\begin{lem}[Optimal PA Position for Single User] \label{thm:optimal_x}
Regarding the optimal PA orientation, the optimal position $\big(x^{(\text{PA})}\big)^*$ is characterized by the following conditions.
\begin{enumerate}[leftmargin=0.4cm]
\item The position is bounded within the interval $[0, x^{(\text{UR})}]$.
\item The optimal coordinate is given by:
\begin{equation}\label{e:optimal_y}
\big(x^{(\text{PA})}\big)^* = x^{(\text{UR})} - d^*,
\end{equation}
where $$d^* = \frac{\alpha_{\text{W}}\left[(\bar{y}^{(\text{UR})} )^2 + (\bar{z}^{(\text{UR})} )^2\right]}{2+\alpha_{\text{A}}\sqrt{\left[(\bar{y}^{(\text{UR})} )^2 + (\bar{z}^{(\text{UR})} )^2\right]}}$$ represents the optimal horizontal offset of the PA relative to the user.
\end{enumerate}
\end{lem}
\begin{proof}(sketch) 
First, if the calculated optimal PA position exceeds $x^{(\text{UR})}$, a corresponding point within the interval $[0, x^{(\text{UR})}]$ can always be identified that maintains an identical wireless propagation distance while incurring lower internal waveguide attenuation. Consequently, the global optimal position must reside within the bounds $[0, x^{(\text{UR})}]$. Then, \eqref{e:optimal_y} can be obtained by solving $\frac{\partial \ln |H|^2}{\partial x^{(\text{PA})}} = 0$. The detailed derivation can be found in Appendix~\ref{sec:AppC}.
\end{proof}

Given the optimal PA orientation and position, we derive the corresponding optimal user antenna polarization. Applying Theorems \ref{Thm:h_WP} and \ref{prop:H_PU}, the electric field distribution at the user location $\bm{\psi}_{k}^{(\mathrm{UR})}$ can be characterized. This field is then reformulated into a spherical coordinate representation as per \eqref{e:E_LSCS}, establishing the incident Jones vector defined in \eqref{e:incident_Jones}. Then, optimal polarization matching is achieved when the user antenna orientation is precisely aligned with the incident wave within the transverse plane orthogonal to the propagation direction, as specified in \eqref{eq:plf_def}, yielding:
\begin{equation}\label{eq:optimal_user_orientation_2D}
\bar{\bm{n}}_{m,n,q,k}^{(r)}
\!=\!
\frac{1}{\mathcal V} \! 
\begin{bmatrix} \!
\left( \!1 \!+\! \dfrac{\beta}{\varrho}\cos\theta \! \right) \! \cos\phi & \!\!\!
\left(\! \dfrac{\beta}{\varrho} \! + \! \cos\theta \! \right) \! \sin\phi \!
\end{bmatrix} \!\!\!
\begin{bmatrix}
\bar{\bm{\vartheta}}\\\bar{\bm{\varphi}}
\end{bmatrix} \!\!,\!\! 
\end{equation}
where $\mathcal{V}$ is a normalization factor ensuring that $\|\bar{\bm{n}}_{m,n,k}^{(r)}\|=1$.

\subsection{Single-PA Two-User Analysis}\label{sec:IVA2}
We now consider a single PA serving two users simultaneously, exploiting the two orthogonal modes ($Q=2$). This scenario captures the essence of mode-division multiplexing in PAsS: the two ports share the same physical location but can radiate independent data streams using different modes. The optimization problem can be written as
\begin{subequations}
\begin{align}
\hspace{-2.2cm}\text{(P1b):}\hspace{3.5cm}& \hspace{-3.5cm}
\max_{ x^{(\text{PA})}, \{\!\xi_{q}^{(\text{PA})}\!\},\{\!\delta_{q}^{(\text{PA})}\!\},\bm{W},\,\{\bm{n}_k^{(r)}\}} 
R_1+R_2, \\
\hspace{-2.2cm}\text{s.t.}\hspace{3.8cm}
& \hspace{-1.3cm} \eqref{e:P1b},\eqref{e:P1e},\eqref{e:P1f},\eqref{e:P1g}.
\label{e:P2}
\end{align}
\end{subequations}

A key observation from the full-wave model is the high spatial selectivity of each mode's radiation pattern, as illustrated in Fig.~\ref{fig:energy_map}. The half-power beamwidth is only about $0.5$ m in the $x$-direction at a height of $3$ m, and the first sidelobe is suppressed by more than $10$ dB. This suggests that when the two users are sufficiently separated (e.g., more than $1$m apart), the inter-user interference caused by side lobes is negligible. We therefore adopt the following spatial separation assumption: the two users are located such that each mode's main lobe illuminates primarily its intended user, and the cross interference is small enough to be omitted. Under this assumption, 1) the optimal orientation for each port is simply to point it directly at its corresponding user, following Lemma \ref{thm:single_angle}. 2) the optimal user antenna polarization can be directly obtained via \eqref{eq:optimal_user_orientation_2D}, once the optimal PA position is determined.

Therefore, the remaining variables in (P1b) to be optimized are $x^{(\text{PA})}$ and $\bm{W}$. Analyzing the sum-rate expression for the two-user configuration shows that the objective is constrained by the following upper bound:
\begin{equation} \label{eq:sumrate_two_user_sharedPA}
\hspace{-0.2cm}R_{\mathrm{sum}}\bigl(x^{(\text{PA})}\bigr)
\!=\! \sum_{q=1}^2\log_2\!\left(1+\frac{P|w^{(p)}_q|^2\left|h_q\bigl(x^{(\text{PA})}\bigr)\right|^2}{ \sigma_q^2}\right) ,
\end{equation}
where $h_q(x^{(\mathrm{PA})})$ denotes the effective channel coefficient for the corresponding user, whose optimal PA location in a single-user configuration is represented by $x_q^*$.

Given the monotonicity of the objective function with respect to $|w^{(p)}_1|^2$ and $|w^{(p)}_2|^2$, the power constraint in \eqref{e:P1b} is necessarily satisfied with $|w^{(p)}_1|^2 + |w^{(p)}_2|^2 = 1$. Substituting this relation into \eqref{eq:sumrate_two_user_sharedPA} and applying the first-order optimality condition $\frac{\partial R_{\mathrm{sum}}\bigl( x^{(\text{PA})} \bigr)}{\partial |w^{(p)}_1|^2} = 0$ yields:
\begin{equation}
\hspace{-0.3cm} \frac{P|h_1( x^{(\text{PA})} )|^2}
{\sigma_1^2 \!+\!P|w^{(p)}_1|^2|h_1( x^{(\text{PA})} )|^2}
\!=\!
\frac{P|h_2( x^{(\text{PA})} )|^2}
{\sigma_2^2\!+\!P(1\!-\!|w^{(p)}_1|^2)|h_2( x^{(\text{PA})} )|^2}.\!\!\!\!\!\!
\label{eq:w1_stationary}
\end{equation}
Solving \eqref{eq:w1_stationary} yields
\begin{equation}
\begin{cases}
    |w^{(p)^*}_1|^2 = \frac{1}{2} +\frac{\sigma_2^2}{2P|h_2( x^{(\text{PA})} )|^2} -\frac{\sigma_1^2}{2P|h_1( x^{(\text{PA})} ) |^2}, \\
    |w^{(p)^*}_2|^2 = \frac{1}{2} +\frac{\sigma_1^2}{2P|h_1( x^{(\text{PA})} )|^2} -\frac{\sigma_2^2}{2P|h_2( x^{(\text{PA})} )|^2}.
    \end{cases}
\end{equation}

Substituting the optimal power allocations $|{w^{(p)}_1}^*|^2$ and $|{w^{(p)}_2}^*|^2$ back into \eqref{eq:sumrate_two_user_sharedPA} reduces the optimization problem to a single variable $x^{(\mathrm{PA})}$. To characterize the solution under the shared-location constraint within the multi-mode system, we perform a second-order Taylor expansion of each rate function $R_q(x^{(\mathrm{PA})})$ around its respective unconstrained local maximizer $x_q^*$, and \eqref{eq:sumrate_two_user_sharedPA} can be reformulated as:
\begin{eqnarray}\label{eq:quadratic_sumrate_sharedPA}
&&\hspace{-0.8cm}R_{\mathrm{sum}}(x^{(\text{PA})}) \approx \sum_{q=1}^{2} \bigg[ R_q(x_q^*)+ R_q^{\prime}(x_q^*) (x^{(\text{PA})} - x_q^*) \\
&&\hspace{3.5cm} +\frac{1}{2}R_q^{\prime\prime}(x_q^*)\bigl(x^{(\mathrm{PA})}-x_q^*\bigr)^2 \bigg].\notag
\end{eqnarray}
Subsequently, by applying $\frac{\partial R_{\mathrm{sum}}\bigl( x^{(\text{PA})} \bigr)}{\partial x^{(\text{PA})}} = 0$ to \eqref{eq:quadratic_sumrate_sharedPA}, we solve (P1b) as formally characterized in Lemma \ref{lem:optimal_position}.

\begin{thm}[Optimal Solution for Two Users]\label{lem:optimal_position}
In (P1b), the optimal PA position $x^*$ and the power allocation $w^{(p)}_1$ and $w^{(p)}_2$ are characterized by the following closed-form approximations: 
\begin{eqnarray}\label{eq:optimal_solution_2usr}
    &&\hspace{-1cm}x^* \approx\frac{R_1^{\prime\prime}(x_1^*)x_1^*+R_2^{\prime\prime}(x_2^*)x_2^*-(R_1^{\prime}(x_1^*)+R_2^{\prime}(x_2^*))}{R_1^{\prime\prime}(x_1^*)+R_2^{\prime\prime}(x_2^*)}, \\
     &&\hspace{-1cm}|{w^{(p)}_1}^*|^2 = \frac{1}{2} +\frac{\sigma_2^2}{2P|h_2( x^* )|^2} -\frac{\sigma_1^2}{2P|h_1( x^* ) |^2}, \\
     &&\hspace{-1cm}|{w^{(p)}_2}^*|^2 =\frac{1}{2} +\frac{\sigma_1^2}{2P|h_1( x^* )|^2} -\frac{\sigma_2^2}{2P|h_2( x^* )|^2},
\end{eqnarray}
where 
\begin{eqnarray*}
\hspace{-0.1cm}R_q^\prime(x_q^*)&\hspace{-0.3cm}=&\hspace{-0.3cm}
-\frac{\sigma_2^2}{2\ln 2}\,\frac{\dfrac{1}{|h_{q^{\prime}}(x_q^*)|^2}
\left.\dfrac{\partial \ln |h_{q^{\prime}}(x^{(\text{PA})})|^2}{\partial x^{(\text{PA})}}\right|_{x^{(\text{PA})}=x_q^*}}
{\dfrac{\sigma_1^2}{|h_q(x_q^*)|^2}+P|w^{(p)}_q(x_q^*)|^2},\\
\hspace{-0.1cm}R_q^{\prime\prime}(x_q^*) &\hspace{-0.3cm} \approx &\hspace{-0.3cm} - \ln 2 \! \left(\! R_q^{\prime}(x_q^*) \! \right)^{\!2} \!-\! R_q^{\prime}(x_q^*) \left.  \!\!\frac{\partial \! \ln \! |h_{q^{\prime}}(x^{(\text{PA})}\!) \! |^{2}}{\partial x^{(\text{PA})}} \!\right|_{ x^{(\text{PA})}=x_q^*},
\end{eqnarray*}
where $q$ and $q^{\prime}$ denote mode 1 and 2, respectively, reflecting the alternating nature between these two operational modes.
\end{thm}
\begin{proof}
(sketch)
Setting the first derivative $\partial R_{\mathrm{sum}}\bigl( x^{(\mathrm{PA})} \bigr) / \partial x^{(\mathrm{PA})} = 0$, we obtain $x^*$. Since the exact second-order behavior is heavily entangled by position-dependent power coupling, we extract the dominant curvature terms under a weak-coupling regime, characterized by 
$|h_{q^\prime}(x_q^*)|^2 \ll 1$ and $\frac{\sigma_{q'}^2 |h_q(x_q^*)|^2}{2P|h_{q'}(x_q^*)|^4}
\gg \frac{\sigma_q^2}{2P|h_q(x_q^*)|^2}.$
to yield a tractable approximation. The detailed derivations of the exact and approximated local derivatives are deferred to Appendix \ref{sec:AppD}.
\end{proof}

The closed-form result reveals that the shared PA location is a compromise between the two users' individual optima, weighted by the curvature of their rate functions. This compromise becomes necessary because the physical position of the PA is common to both ports.

\subsection{General Multi-PA Multi-User Optimization}\label{sec:IVC}

Building upon the insights from the single-PA analysis, we now tackle the general problem (P1) for systems with multiple waveguides and PAs. The key observation from the full-wave model is that each guided mode radiates a spatially selective beam. Therefore, a single PA equipped with $Q$ modes can simultaneously serve up to $Q$ users that are sufficiently separated in space, with each mode's main lobe illuminating its intended user. This motivates us to partition the user set into groups of size $Q$, where each group will be served by one dedicated PA. The optimal orientation of each PA port then follows directly from Lemma \ref{thm:single_angle}, and once a PA is paired with a specific user group, the optimal PA position can be obtained from Theorem \ref{lem:optimal_position}.

Consequently, for the general problem (P1) the remaining tasks are:
\begin{itemize}[leftmargin=0.5cm]
    \item User grouping: partition the $K$ users into $J=K/Q$ disjoint groups, each containing $Q$ users.
    \item PA assignment: assign each of the $MN$ PAs to one of these user groups.
    \item Precoding optimization: design the digital precoder that allocates power among modes and PAs.
\end{itemize}
We define $\mathcal{I}=\{1,\ldots,MN\}$ as the PA index set and $\mathcal{K}=\{1,\ldots,K\}$ as the user index set. Let $\mathcal{P}_j = \{j_1, \dots, j_Q\}$ denote the $j$-th $Q$-user group, where $j_q $ represents the $q$-th user in this $j$-th group, and $J = K/Q$ is the total number of groups.
To characterize the PA-to-group topology, we introduce a binary assignment matrix $\bm{X} \in \{0,1\}^{MN \times J}$. $X_{i,j}=1$ indicates that the $i$-th PA is assigned to serve group $\mathcal{P}_j$, and $X_{i,j}=0$ otherwise.

\subsubsection{User grouping}
Because each PA serves its assigned group via spatially separated beams, users within the same group should be physically close to each other; otherwise, the required beam separation forces a larger propagation distance and reduces the channel gain. Hence we use the sum of pairwise squared Euclidean distances as the grouping metric:
\begin{equation}\label{eq:leftover_pairing}
\arg\min_{\{\mathcal P_j \}}
\!\sum_{\mathcal P_j =\{j_q\}}
\!\!\Big(\left|x_k^{(\text{UR})}\!-\!x_{k'}^{(\text{UR})}\right|^2\!+\!\left|y_k^{(\text{UR})}\!-\!y_{k'}^{(\text{UR})}\right|^2\Big).
\end{equation}
Exhaustive search over all $\frac{K!}{(Q!)^{K/Q} (K/Q)!}$ groupings is prohibitive for large $K$. We therefore propose a geometry-guided pairing strategy that leverages the physical topology of the waveguide array. Specifically, since the waveguides are displaced along the $y$-axis, a user's $y$-coordinate naturally determines its primary service waveguide. Each user is therefore initially associated with its proximate waveguide. Within these waveguide-specific clusters, users are sorted by their $x$-coordinates, and adjacent users are preferentially paired.

In cases where a cluster size is odd, any unmatched users are assigned to a residual set and paired according to \eqref{eq:leftover_pairing}. This approach remains computationally tractable, as the cardinality of the residual set is sufficiently small to permit an exhaustive search. This localized sorting rule effectively clusters users for joint service by the same PA. By pruning the search space, the strategy ensures that the final grouping is determined via a low-complexity exhaustive search over a restricted candidate set, thereby maintaining overall computational efficiency.

\subsubsection{PA assignment}
After user grouping, we proceed to the matching stage between PAs and user groups, which is formulated as:
\begin{subequations}\label{eq:P3_joint_new}
\begin{align}
\arg\max_{\bm X}\quad &
\sum_{i\in\mathcal I}\sum_{j=1}^{J} X_{i,j}R^{(2)}_{i,j}(x^{(\text{PA})}) \label{eq:P3_joint_new_obj}\\
\text{s.t.}\quad
& \sum_{j=1}^{J} X_{i,j} = 1,\label{eq:P3_joint_new_row}\\
& \sum_{i\in \mathcal{I}} X_{i,j} \geq 1 \quad \text{or}  \quad MN < J,\label{e:P1c_f}
\end{align}
\end{subequations}
where \eqref{eq:P3_joint_new_row} guarantees full antenna utilization, while \eqref{e:P1c_f} ensures each group is allocated at least one antenna. 
Here, $R^{(2)}_{i,j}$ is the maximum achievable sum-rate when PA $i$ serves group $j$. It generalizes the two-user sum-rate in \eqref{eq:quadratic_sumrate_sharedPA} by including inter-user interference from other PAs:
\begin{equation}\label{eq:R2}
\hspace{-0.2cm}
R_{i,j}^{(2)}(x^{(\text{PA})})\!=\!\sum_{q=1}^{2}\log_2\!\left(\!1\!+\!\frac{P|w^{(p)}_{i,j_q}|^2\left|h_{i,j_q}\bigl(x^{(\mathrm{PA})}\bigr)\right|^2}{\sigma_{j_q}^2 \!+ \!{\mathcal N}_{i,j_q}}\!\right),
\end{equation}
where
\begin{equation*}
{{\mathcal N}}_{i,j_q}
=P\!\!\!\!\!\!\!\sum_{\substack{(i',j')\neq(i,j)}}
\!\!\!\!\!\!\!X_{i',j'}\,|w^{(p)}_{i',j_q}|^2
\left|h_{i',j_q}\bigl(x_{i',j'}^{(\mathrm{PA})}\bigr)\right|^2 
\end{equation*}
is inter-user interference, and $|w^{(p)}_{i,j_{q}}|^2$ denotes the local split coefficient from this PA to the $q$-th user in pair $j$.  The optimal sum-rate $R_{i,j}^{(2)}(x^{(\text{PA})})$ is obtained by substituting $\sigma_q^2$ in \eqref{eq:sumrate_two_user_sharedPA} with the effective noise-plus-interference term $\sigma_{j_q}^2 + {\mathcal{N}_{i,j_q}}$. 

This formulation constitutes a deformed rectangular linear assignment problem, which can be solved using the Hungarian algorithm. In cases where $MN < J$, the problem is addressed by introducing $J - MN$ dummy antennas with zero sum-rate contributions ($R^{(2)}_{\tilde{i},j} = 0$). Conversely, when $MN > J$, $MN - J$ dummy user groups with $R^{(2)}_{i,\tilde{j}} = 0$ are incorporated to facilitate the assignment. Then, the Hungarian algorithm is employed to uniquely assign one PA to each user group. 

The remaining unassigned PAs, denoted by the set $\tilde{\mathcal{I}}$, are subsequently allocated via a greedy heuristic based on their marginal contributions to the objective function in \eqref{eq:P3_joint_new_obj}. Specifically, for an unassigned PA $i_l \in \tilde{\mathcal{I}}$ and the user group $j \in \{1, \dots, J\}$, let $\tilde{\bm X}^{(t)} = \bm X^{(t)} + \bm {\mathcal{X}}_{i_l,j}$ denote the tentative assignment matrix, where $\bm X^{(0)}$ represents the initial assignment matrix after applying the Hungarian algorithm and $\bm {\mathcal{X}}_{i_l,j}$ is an $MN \times J$ single-entry matrix with $1$ at the $(i_l,j)$-th position and $0$ elsewhere. The marginal gain is defined as the exact increment of the objective value:
\begin{equation}\label{eq:marginal_gain_compact}
\Delta_{i,j}(\tilde{\bm X}^{(t)}) \triangleq \sum_{i\in\mathcal I}\sum_{j=1}^{J} \tilde{X}_{i,j} \tilde{R}^{(2)}_{i,j} - \sum_{i\in\mathcal I}\sum_{j=1}^{J} X_{i,j}R^{(2)}_{i,j},
\end{equation}
where $\tilde{R}^{(2)}_{i,j}$ is the achievable rate evaluated under the tentative assignment $\tilde{\bm X}^{(t)}$. In each iteration, the pair yielding the maximal marginal gain is selected for the next assignment. Then, 
the assignment state is then updated as $\bm X^{(t)} \leftarrow \bm X^{(t+1)} + \bm {\mathcal{X}}_{i_l^*, j^*}$, and the allocated antenna $i_l$ is removed from the available set $\tilde{\mathcal{I}}$. This procedure repeats iteratively until all PAs are allocated.

Once the assignment is fixed, the optimal PA position for each group follows directly from Theorem \ref{lem:optimal_position}, and the optimal orientations of the PA ports and user antennas are given by Lemma \ref{thm:single_angle} and \eqref{eq:optimal_user_orientation_2D}, respectively.

\subsubsection{Precoding optimization}\label{sec:fp_mo_ao}
With spatial deployment fixed, we now optimize the digital precoder. 
The power allocation coefficients $w^{(p)}_{i,j_1}$ and $w^{(p)}_{i,j_2}$ from Section \ref{sec:IVA2} are arranged into a sparse precoding matrix $\bm{W}^{(p)} \in \mathbb{C}^{MN \times K}$. This matrix is structured such that $[\bm{W}^{(p)}]_{i, 2j-1} = w^{(p)}_{i,j_1}$ and $[\bm{W}^{(p)}]_{i, 2j} = w^{(p)}_{i,j_2}$ if the $i$-th PA is assigned to the $j$-th user pair, while the remaining entries in the $i$-th row are zero. Consequently, each row of $\bm{\mathcal{F}}$ contains exactly two non-zero entries, characterizing the localized power distribution of the $i$-th PA between the two users within its assigned group.

We further introduce a matrix $\bm{G} = [\bm{g}_1, \ldots, \bm{g}_{MN}] \in \mathbb{C}^{QM \times MN}$ to characterize the power allocation from the access point to each antenna element. The $i$-th column $\bm{g}_i$ specifies the mode-domain excitation allocated to the $i$-th power amplifier (PA). Then, the global system precoder $\bm{W}$ as defined in \eqref{e:s} can be factorized as $\bm{W} = \bm{G}\bm{W}^{(p)}$ or $\bm w_k=\bm G \bm w^{(p)}_k$, where $\bm w^{(p)}_k$ is the $k$-th column of $\bm{W}^{(p)}$. 

Therefore, the non-convex sum-rate maximization problem can be addressed via an iterative framework that integrates FP. By applying the quadratic transform, an equivalent FP reformulation is derived as follows:
\begin{subequations}\label{e:optimize_G}
\begin{align}
&\hspace{-1cm} \max_{\bm{G},\bm{c}_1,\bm{c}_2}  \mathcal{L} = \sum_{k=1}^{K} (1+c_{1,k}) \Bigg[ 2\Re\big\{c_{2,k}\sqrt{P}\,\bm{h}_k \bm{G}\bm w^{(p)}_k\big\} \label{eq:fp_obj_revised} \\
&\hspace{-1cm}  - |c_{2,k}|^2 \! \bigg( \! \sigma^2 \! + \! P \!\sum_{k'=1}^{K} \! \big|\bm{h}_k \bm{G}\bm w^{(p)}_{k^\prime}\big|^2 \! \bigg) \! \Bigg] \! + \! \sum_{k=1}^{K} \bigl( \log(1 \! + \! c_{1,k}) \! - \! c_{1,k} \bigr)  \notag  \\
\text{s.t.} \quad & c_{2,k} \in \mathbb{C}, \quad c_{1,k} > 0,  \label{eq:fp_aux_constraints} \\
& \operatorname{tr}\bigl(\bm{G}\bm{W}^{(p)}(\bm{W}^{(p)})^{\mathsf H}\bm{G}^{\mathsf H}\bigr) \le 1. \label{eq:fp_phys_constraints}
\end{align}
\end{subequations}
where $\Re\{\cdot\}$ denotes the real part of a complex value, while $\bm c_1=[c_{1,1},\ldots,c_{1,K}]$ and $\bm c_2=[c_{2,1},\ldots,c_{2,K}]$ represent the auxiliary variables. For a given $\bm G$, the optimal values for these auxiliary variables admit the following closed-form updates:
\begin{align}
c_{1,k}^*
&=\frac{P\big|\bm h_k\bm G\bm w^{(p)}_k\big|^2}
{\sigma^2+P\sum_{k'\neq k}\big|\bm h_k\bm G\bm w^{(p)}_{k^\prime}\big|^2},\label{eq:u_update_revised}\\
c_{2,k}^*
&=\frac{\sqrt{P}\,\bm h_k\bm G\bm w^{(p)}_k}
{\sigma^2+P\sum_{k'=1}^{K}\big|\bm h_k\bm G\bm w^{(p)}_{k^\prime}\big|^2}.
\label{eq:v_update_revised}
\end{align}
Under the generalized power constraint, \eqref{e:optimize_G} reduces to a convex quadratic program. Then the optimal precoding matrix $\bm G$ is then characterized by the Karush-Kuhn-Tucker optimality condition: 
\begin{equation*}
\begin{split}
&\left( \sum_{k=1}^{K} \mu_k\,\bm h_k^{\mathsf{H}}\bm h_k + \chi\bm I \right) \bm G \bigl[ \bm{W}^{(p)} (\bm{W}^{(p)})^{\mathsf{H}} \bigr] \\
&\quad = \sqrt{P} \sum_{k=1}^{K} (1 + c_{1,k}) c_{2,k}^{*} \bm h_k^{\mathsf{H}} (\bm w^{(p)}_k)^{\mathsf{H}},
\end{split}
\end{equation*}
where $\mu_k \triangleq P(1+c_{1,k})|c_{2,k}|^2$, and $\chi$ denotes the Lagrange multiplier associated with the generalized power constraint. If $\bm{W}^{(p)} (\bm{W}^{(p)})^ {\mathsf{H}}$ is rank-deficient, the Moore-Penrose pseudoinverse $[\bm{W}^{(p)} (\bm{W}^{(p)})^ {\mathsf{H}}]^\dagger$ is used to ensure a stable solution. In practice, $\chi$ is obtained via bisection over the monotonic function $\operatorname{tr}\!\bigl(\bm G(\chi)\bm{W}^{(p)}(\bm{W}^{(p)})^{\mathsf{H}}\bm G(\chi)^{\mathsf{H}}\bigr)$ to satisfy the power constraint.

By embedding the FP framework into the structured precoder factorization $\bm W=\bm G \bm{W}^{(p)}$, the proposed method strictly preserves the sparse user grouping and local power-sharing patterns encoded in $\bm{W}^{(p)}$. This structure-guided formulation confines the iterative optimization exclusively to $\bm G$. Consequently, unlike conventional full-dimensional designs, it enables closed-form auxiliary updates and an efficient KKT-based solution for $\bm G$, significantly reducing computational complexity while aligning with the PASS architecture.

\section{Numerical Results}\label{sec:V}
This section evaluates the performance of the proposed multi-mode PASS framework through numerical simulations. These results validate the preceding analytical derivations and demonstrate the system's efficacy in terms of polarization characteristics and sum-rate optimization. 
Unless otherwise specified, simulations are conducted at a carrier frequency of $100$ GHz. The rectangular dielectric waveguides have cross-section $a \times b = 3 \times 2$ mm$^2$; the core refractive index is $2.0$. 
The waveguide and atmospheric attenuation coefficients are set to $\alpha_{\text{W}} = 0.08$ dB/m and $\alpha_{\text{A}} = 0.05$ dB/m, respectively. 
The total transmit power is $P = 10$ W, with its logarithmic equivalent $P_{\mathrm{dBW}} \triangleq 10 \log_{10}(P)$. 
Noise power is $\sigma^2 = -26$ dBW, and the minimum inter-PA spacing is half wavelength.

\subsection{Single-PA Performance}
We begin with a single PA serving one or two users to validate the closed-form results and the fundamental trade-offs derived in Section \ref{sec:IV}.

\subsubsection{Polarization matching}
Fig.~\ref{f:polarization_map} shows the channel gain $|H|^2$ as a function of PA orientation (Fig.~\ref{f:polar_tx_rotate}) and user antenna polarization (Fig.~\ref{f:polar_rx_rotate}) for a single user at $(5.5,0,0)$ and a PA at $(5,0,3)$. In Fig.~\ref{f:polar_tx_rotate}, the receiver polarization is fixed to vertical, and the PA orientation is swept. The gain peaks exactly at the orientation predicted by Lemma~\ref{thm:single_angle}, confirming the analytical result. 
The maximum gain is only $0.02$ because the receiver polarization is not matched. In Fig.~\ref{f:polar_rx_rotate}, the PA is optimally pointed at the user, and the receiver polarization is varied. Now the gain reaches approximately $0.6$ when the receive Jones vector aligns with the incident field, as given by \eqref{eq:optimal_user_orientation_2D}. Outside this matched condition, the gain drops sharply, underscoring the importance of polarization alignment.

\begin{figure}[!tb]
\centering
    \begin{subfigure}[!tb]{0.49\linewidth}
        \centering
        \includegraphics[width=1\linewidth]{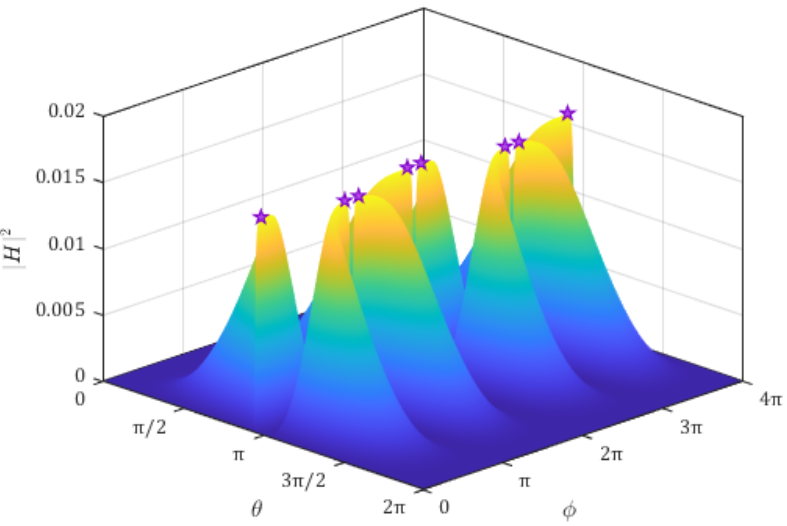}
        \caption{}\label{f:polar_tx_rotate}
    \end{subfigure}
    \begin{subfigure}[!tb]{0.49\linewidth}
        \centering
        \includegraphics[width=1\linewidth]{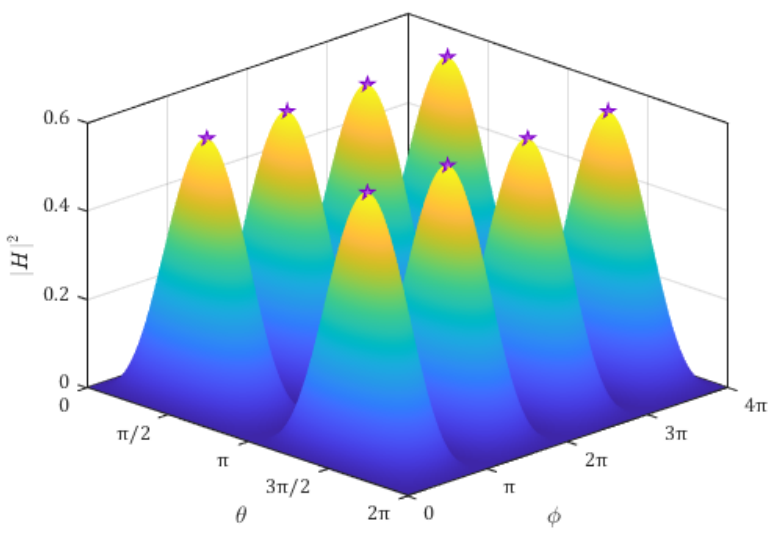}
        \caption{}\label{f:polar_rx_rotate}
    \end{subfigure}
    
\caption{Single-port single-user channel gain versus transmit and receive orientations: (a) PA rotation with fixed receiver orientation; (b) receiver rotation with PA fixed at its optimal orientation. The optimal PA orientation derived in Lemma \ref{thm:single_angle} and the corresponding user polarization orientation derived in \eqref{eq:optimal_user_orientation_2D} are also labeled in the figure. }
\label{f:polarization_map}
\end{figure}

\subsubsection{Single PA serving two users}
Fig.~\ref{fig:single_curve} plots the sum rate versus PA position $x^{(\text{PA})}$ for two user pairs: a narrow pair at $(4.5,0,0)$ and $(5.5,0,0)$, and a wide pair at $(3,0,0)$ and $(7,0,0)$. Both configurations are symmetric about $x=5$ m. The sum rate increases as the PA moves toward the users, but after passing the midpoint it decreases because the wireless path loss grows and the waveguide attenuation continues to increase. The optimal shared position lies near the midpoint, with a slight feed-side bias due to waveguide loss. The narrow pair achieves a higher peak rate because both users can be kept close simultaneously, whereas the wide pair forces a larger average distance. Multi-mode transmission (two modes active) consistently delivers about $1.5$ times the rate of single-mode time-division access (TDMA). The gain is less than the theoretical factor $2$ because the PA must compromise between the two users' optimal positions.

\begin{figure}[!tb]
    \centering
    \includegraphics[width=0.8\linewidth]{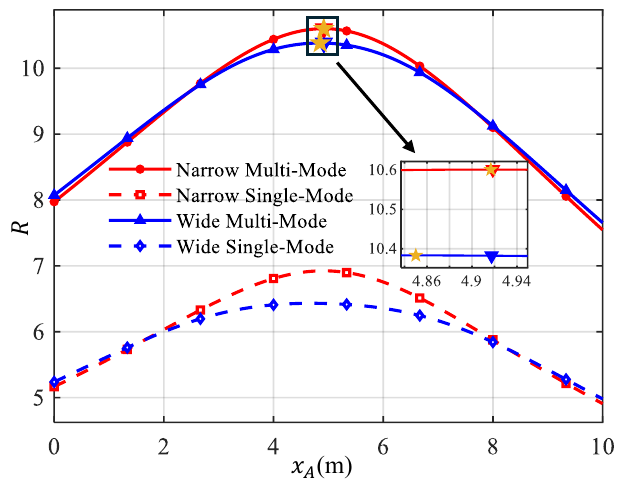}
    \caption{Sum-rate versus antenna position $x_A$ for narrow-spacing and wide-spacing user groups in a single-PA setup. }
    \label{fig:single_curve}
\end{figure}

\subsubsection{Outage probability}
Fig.~\ref{fig:outage} shows the outage probability (defined as the event that either user's rate falls below $1$ bps/Hz) versus transmit power for different atmospheric absorption coefficients $\alpha_{\text{A}}$. The multi-mode scheme is compared with a single-mode TDMA baseline. Higher $\alpha_{\text{A}}$ shifts the curves to the right, requiring more power to achieve a given outage. Nevertheless, multi-mode transmission provides a substantial power saving: at $10^{-4}$ outage, it needs $1.9$ dB less power than the single-mode baseline (both with $\alpha_{\text{A}}=0.05$ dB/m). This diversity gain compensates even for a threefold increase in atmospheric absorption, highlighting the robustness of mode-division multiplexing.

\begin{figure}
    \centering
    \includegraphics[width=0.8\linewidth]{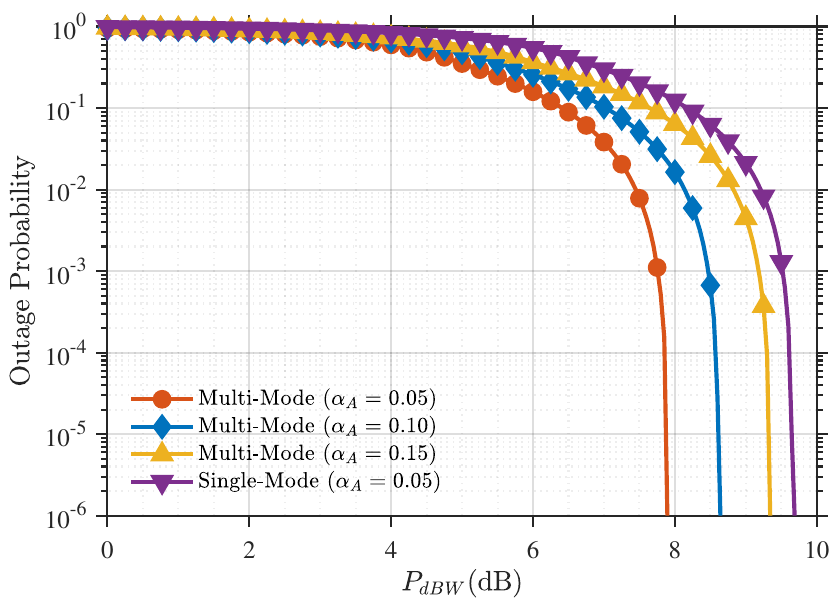}
    \caption{Outage probability versus transmit power $P_{dBW}$ under different atmospheric absorption coefficients $\alpha_A$.}
    \label{fig:outage}
\end{figure}

\begin{figure}
    \centering
    \includegraphics[width=0.8\linewidth]{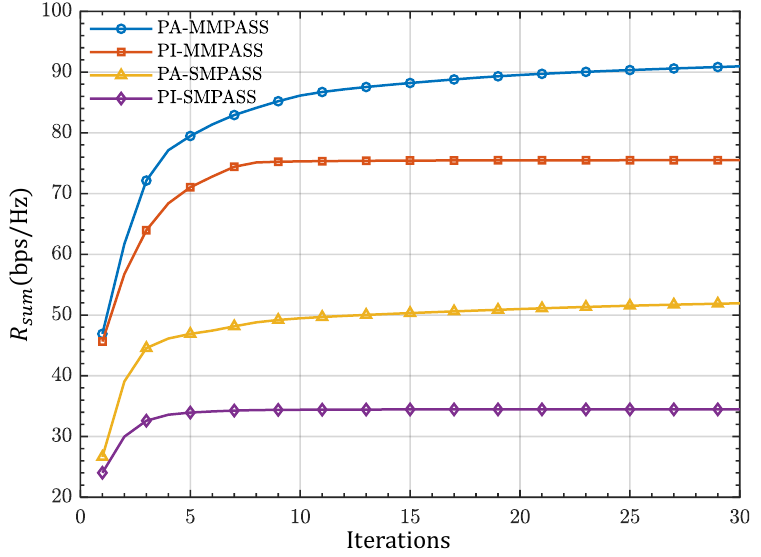}
    \caption{Convergence behavior of the proposed joint optimization framework for different system configurations.}
    \label{fig:MM_SM_PA_PI}
\end{figure}

\subsection{MMPASS Performance}

We now evaluate the performance of the general MMPASS framework by incorporating its complete electromagnetic characteristics. This analysis reveals two synergistic mechanisms that enhance the sum-rate: 1) spatial multiplexing gains inherent in multi-mode transmission, and 2) optimal polarization alignment to minimize propagation loss. To validate these joint benefits, we investigate the following system configurations:
\begin{itemize}[leftmargin = 0.5cm]
    \item \textbf{PI-SMPASS:} Polarization-Ignorant Single-mode PASS;

    \item \textbf{PA-SMPASS:} Polarization-Aware Single-mode PASS;

    \item \textbf{PI-MMPASS:}Polarization-Ignorant Multi-mode PASS

    \item \textbf{PA-MMPASS:} Polarization-Aware Multi-mode PASS.

     \item \textbf{DP-MMPASS:} Discrete-Polarization Multi-mode PASS. This configuration achieves partial polarization matching by selecting the receiver's polarization state from a predefined discrete set. Concurrently, the PA is assumed to be perfectly aligned with the optimal radiation direction.
\end{itemize}

\subsubsection{Convergence behavior}
Fig.~\ref{fig:MM_SM_PA_PI} shows the sum rate versus iteration count for the joint FP and polarization update algorithm. PA-MMPASS converges to about $91$ bps/Hz, while PI-MMPASS, PA-SMPASS, and PI-SMPASS reach $74$, $51$, and $34$ bps/Hz, respectively. The gains decompose as follows: polarization awareness alone improves rate by $23\%$ (from $74$ to $92$) in multi-mode and by $50\%$ (from $34$ to $51$) in single-mode; multi-mode alone improves rate by $78\%$ (from $51$ to $91$) under polarization awareness. 
Overall, the MMPASS framework achieves up to a $167\%$ increase in spectral efficiency compared with single-mode PASS.

\subsubsection{Sum rate vs. transmit power}
Fig. \ref{fig:r_t} presents the achievable sum-rate versus transmit power with the proposed three-stage optimization framework. We have the following main observations.

\begin{figure}
    \centering
    \includegraphics[width=0.85\linewidth]{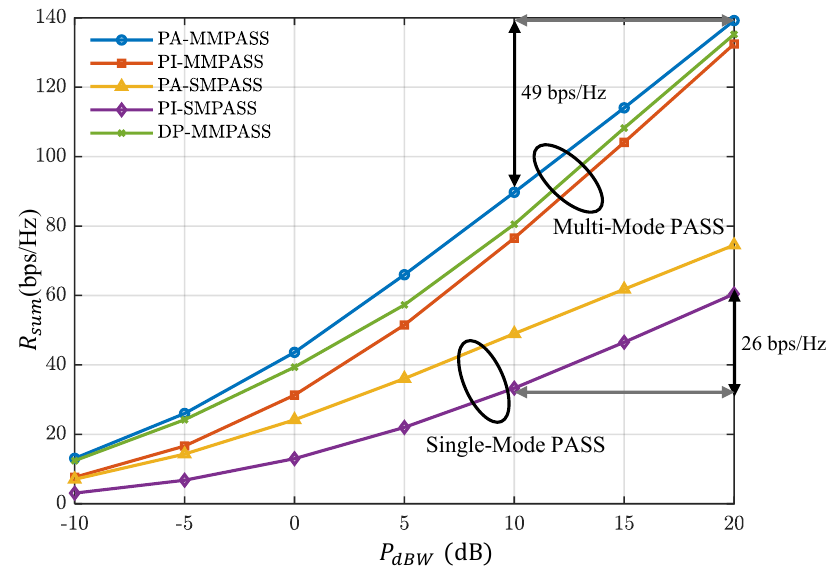}
    \caption{Achievable sum-rate $R_{sum}$ versus transmit power $P_{dBW}$ for the different PASS schemes.}
    \label{fig:r_t}
\end{figure}

\begin{itemize}[leftmargin= 0.4cm]
    \item All multi-mode schemes consistently outperform their single-mode counterparts across the entire power range. For instance, at a transmit power of $20$ dBW, PA-MMPASS achieves approximately $139$ bps/Hz, whereas PA-SMPASS and PI-SMPASS yield only $74$ and $60$ bps/Hz, respectively.
    \item The slope in the high-SNR regime is about $49$ bps/Hz per $10$ dB for multi-mode, nearly double the $26$ bps/Hz per $10$ dB for single-mode, confirming the higher multiplexing gain. Within the multi-mode category, PA-MMPASS consistently provides the superior performance compared with PI-MMPASS. This indicates that polarization matching yields non-negligible gains even after the topology is fixed by Algorithm 1.
    \item Within each respective transmission mode, Polarization-Aware configurations consistently outperform their Polarization-Ignorant counterparts. Specifically, at $P = 20$ dBW, PA-MMPASS provides a gain of approximately $7$ bps/Hz over PI-MMPASS, while the performance gap between PA-SMPASS and PI-SMPASS reaches $14$ bps/Hz. These results indicate that polarization matching yields non-negligible spectral efficiency improvements even after the physical system topology has been fixed.
   \item Furthermore, the Jones vector of user for PP-MMPASS scheme $\bar{\bm n}^{(r)}=(\cos \bar{\phi}^{(r)}, \sin \bar{\phi}^{(r)})$, utilizing a discrete angular codebook constructed by uniformly sampling  $\bar{\phi}^{(r)} \in \left\{0,\frac{\pi}{9},\frac{2\pi}{9},\ldots,\frac{17\pi}{9}\right\}$, closely tracks the performance of PA-MMPASS. At $P = 20$~dBW, the spectral efficiency gap between the two is a mere $4$~bits/s/Hz. This suggests that discrete polarization control can capture the majority of the gains from continuous optimization while maintaining lower implementation complexity.
\end{itemize}

\subsubsection{Scalability of MMPASS}
Fig.~\ref{WG_Sumrate} shows the sum rate as a function of the number of waveguides $M$ and the number of PAs per waveguide $N$, with $K=24$ users. The sum rate increases monotonically with both $M$ and $N$ for all schemes, as more waveguides provide parallel spatial channels and more PAs offer finer beamforming control. The multi-mode surfaces have a steeper gradient than the single-mode ones, indicating that MMPASS better exploits additional hardware resources. PA-MMPASS maintains the highest rate across all $(M,N)$ combinations, followed by PI-MMPASS, PA-SMPASS, and PI-SMPASS.

\begin{figure}
    \centering
    \includegraphics[width=0.8\linewidth]{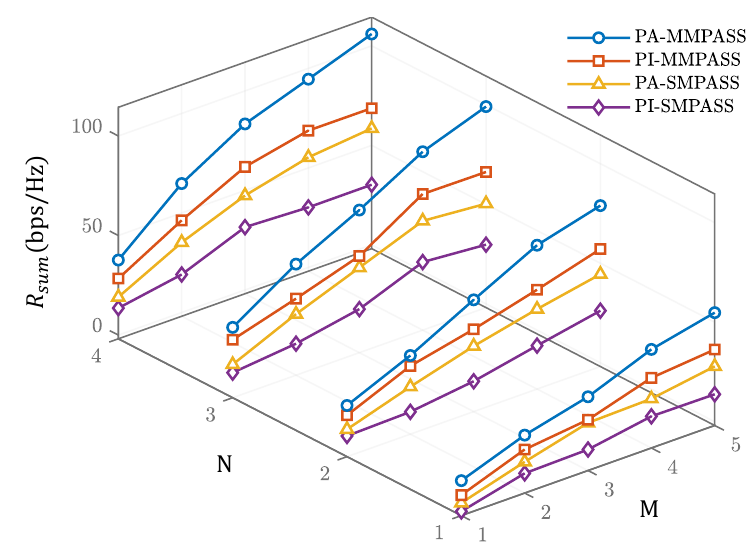}
   \caption{Achievable sum-rate $R_{sum}$ versus the numbers of waveguides $M$ and PAs $N$ for different PASS models.}
    \label{WG_Sumrate}
\end{figure}

Fig.~\ref{fig:usr_sumrate} examines the sum rate versus the number of users $K$ (with $M=4$, $N=3$). All schemes exhibit a monotonic increase, benefiting from multi-user diversity. PA-MMPASS again achieves the highest rate, while DP-MMPASS closely tracks it and clearly surpasses PI-MMPASS as $K$ grows, demonstrating that discrete polarization control becomes increasingly valuable in dense user scenarios. The multi-mode framework maintains roughly double the spectral efficiency of single-mode PASS over the entire range of $K$.

\begin{figure}
    \centering
    \includegraphics[width=0.8\linewidth]{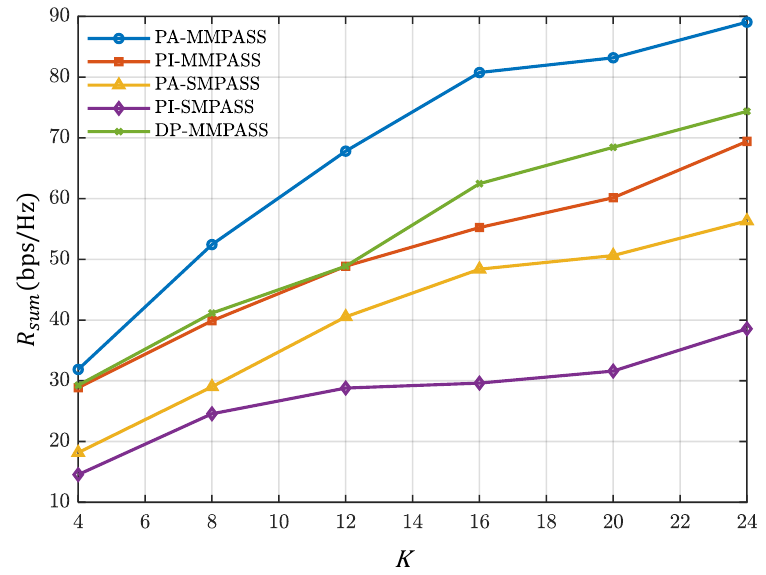}
    \caption{Achievable sum-rate $R_{sum}$ versus the number of users $K$ for different PASS models.}
    \label{fig:usr_sumrate}
\end{figure}

\subsubsection{Algorithm convergence comparison}
Finally, Fig.~\ref{fig:algocomp} compares the convergence speed of the proposed algorithm, which combines FP-based precoder optimization with a closed-form update of the polarization matching vector, with two benchmarks: a Riemannian manifold method that optimizes the receive polarization vectors on the unit sphere (using conjugate gradient on the manifold) \cite{yu2016alternating} and a signal-to-leakage-plus-noise ratio (SLNR) precoder \cite{sadek2007leakage} that treats interference as noise. 
The proposed method converges within $15$ iterations to a sum rate of about $92$ bps/Hz, significantly faster than the manifold-based approach (which requires $40+$ iterations) and achieves a $12$ bps/Hz higher final rate than SLNR. This confirms that the closed-form polarization update, combined with FP precoding, provides both fast convergence and near-optimal performance.

\begin{figure}
    \centering
    \includegraphics[width=0.8\linewidth]{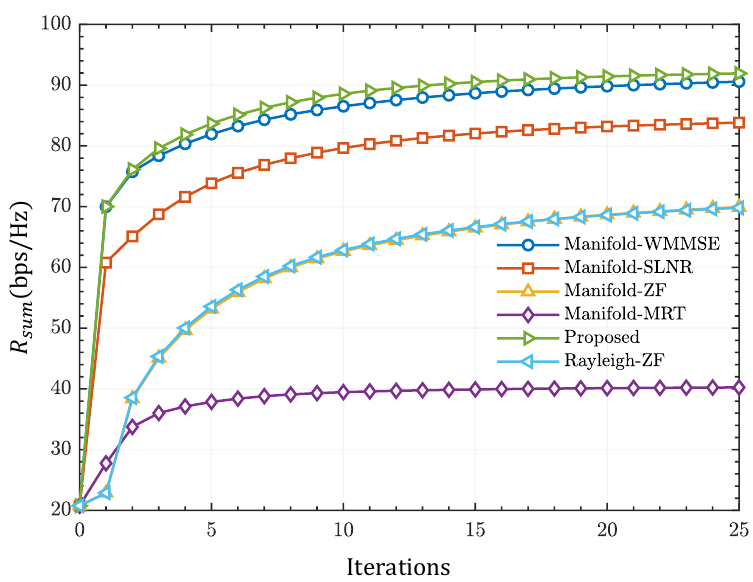}
    \caption{Achievable sum-rate versus the number of iterations for the proposed algorithm against a  manifold-based approach \cite{yu2016alternating} and an SLNR-maximization approach \cite{sadek2007leakage}.}
    \label{fig:algocomp}
\end{figure}

\section{Conclusion}
This paper has introduced a physically complete, polarization-aware electromagnetic modeling framework for MMPASS, replacing the black-box CMT modeling with a first-principles characterization that captures the full spatial and polarization degrees of freedom inherent in guided-wave radiation. By analytically deriving optimal design rules for PA placement, orientation, and polarization matching, and by developing a scalable optimization architecture for multi-user scenarios, this work establishes a theoretical foundation that elevates MMPASS from a conceptually promising architecture to a quantitatively tractable engineering paradigm.

The impact of this framework extends beyond the specific results presented. From a theoretical perspective, the full-wave model opens the door to studying higher-order mode interactions, mutual coupling effects, and more complex waveguide geometries, all of which were previously inaccessible due to the limitations of coupled-mode approximations. 
The closed-form insights into the trade-off between waveguide attenuation and free-space path loss can inform the development of automated deployment strategies for dense antenna infrastructures, particularly in indoor millimeter-wave and terahertz environments. 
On the practical side, the finding that discrete polarization control achieves near-optimal performance suggests a low-complexity implementation path, bringing MMPASS closer to real-world prototyping and standardization. More broadly, the integration of polarization awareness into the system design framework provides a blueprint for future work that combines MMPASS with complementary technologies, such as reconfigurable antennas and AI-driven beam management, where the rich modal and polarization resources can be jointly exploited to meet the demanding requirements of next-generation wireless networks.

\appendices

\section{Proof of Theorem~\ref{thm:E_radiation}}\label{sec:AppA}

This appendix presents the detailed derivation of the PA's radiation pattern in \eqref{eq:E_field_rigorous}. Let $S$ denote the aperture interface between the waveguide region and free space, and let $\hat{\bm n}$ be the unit normal pointing from the aperture of PA side to the free-space side. We denote by $(\cdot)^{(\text{PA})}$ the original fields on the PA aperture, and by $(\cdot)^{(\text{FS})}$ the limiting values of the fields on the free-space side of $S$.

Starting from Maxwell's equations:
\begin{align*}
\nabla \times \bm E &= -j\omega \bm B - \bm M_s, \\
\nabla \times \bm {\mathcal H} &= \bm J_s + j\omega \bm D, \\
\nabla \cdot \bm D &= \rho_e, \\
\nabla \cdot \bm B &= \rho_m,
\end{align*}
where $\rho_e$ and $\rho_m$ represent the electric and equivalent magnetic charge densities, respectively; $\bm D=\varepsilon \bm E$ and $\bm B=\mu \bm {\mathcal{H}}$ are the electric and magnetic flux densities. $\bm J_s$ and $\bm M_s$ denote the equivalent surface electric and magnetic current densities supported on $S$.

By applying Ampère's Circuital Law and Faraday's Law to an infinitesimal contour $\partial\Sigma$ spanning a surface $\Sigma$ across the aperture $S$, we obtain
\begin{align*}
\oint_{\partial\Sigma} \bm E \cdot d\bm l &= - \iint_{\Sigma} (j\omega \bm B + \bm M_s)\cdot \bm n_{\Sigma}\, dS, \\
\oint_{\partial\Sigma} \bm {\mathcal H} \cdot d\bm l &= \iint_{\Sigma} (\bm J_s + j\omega \bm D)\cdot \bm n_{\Sigma}\, dS. 
\end{align*}
Let the contour have tangential edge vector $\bm L=L\bm t$ and height $h$, with surface normal $\bm n_{\Sigma}=\hat{\bm n}\times \bm t$. Taking the limit $h\to 0$, the flux terms vanish, and the contour integrals reduce to
\begin{align*}
(\bm {\mathcal H}^{(\text{FS})}-\bm {\mathcal H}^{(\text{PA})})\cdot \bm L &= \bm J_s \cdot (\hat{\bm n} \times\bm L ), \\
(\bm E^{(\text{FS})}-\bm E^{(\text{PA})})\cdot \bm L &= - \bm M_s \cdot (\hat{\bm n} \times \bm L).
\end{align*}
Since $\bm L$ is arbitrary on the aperture plane, the tangential boundary conditions follow as
\begin{align*}
\hat{\bm n}\times(\bm {\mathcal H}^{(\text{FS})}-\bm {\mathcal H}^{(\text{PA})}) &= \bm J_s, \\
\hat{\bm n}\times(\bm E^{(\text{FS})}-\bm E^{(\text{PA})}) &= -\bm M_s. 
\end{align*}
Invoking Love's equivalence principle, we suppress the fields on one side of the aperture and replace the original aperture by equivalent surface currents that reproduce the radiated fields in free space. The equivalent currents are given by
\begin{align}
\bm J_s &= \hat{\bm n}\times \bm {\mathcal H}^{(\text{PA})}, \label{eq:J_s_appB}\\
\bm M_s &= -\hat{\bm n}\times \bm E^{(\text{PA})}. \label{eq:M_s_appB}
\end{align}
Next, we construct the radiated fields in the exterior homogeneous region $(\varepsilon,\mu)$ from the equivalent surface currents $(\bm{J}_s,\bm{M}_s)$. We introduce the magnetic vector potential $\bm A(\bm \psi)$ and electric vector potential $\bm F(\bm \psi)$, together with the electric scalar potential $\Phi_e$ and magnetic scalar potential $\Phi_m$, such that
\begin{align}
    \bm E &= -j\omega \bm A -\nabla\Phi_e -\frac{1}{\varepsilon}\nabla\times\bm F, \label{eq:E_pot_def}\\
    \bm {\mathcal{H}} &= \frac{1}{\mu}\nabla\times\bm A -j\omega\bm F -\nabla\Phi_m.
\end{align}
Imposing the Lorenz gauge conditions $\nabla\cdot\bm A = -j\omega\mu\varepsilon\Phi_e$ and $\nabla\cdot\bm F = -j\omega\mu\varepsilon\Phi_m$, the potentials satisfy the vector Helmholtz equations:
\begin{equation} \label{eq:helmholtz_pair}
    (\nabla^2+\varrho^2)\bm A = -\mu\bm{J}_s, \quad 
    (\nabla^2+\varrho^2)\bm F = -\varepsilon\bm{M}_s,
\end{equation}
where $\varrho = \omega\sqrt{\mu\varepsilon}$. The solution is obtained using the outgoing-space Green's function $G(\bm{\psi},\bm p) =\frac{e^{-j\varrho\|\bm{\psi}-\bm p\|}}{(4\pi \|\bm{\psi}-\bm p\|)} $ :
\begin{equation} \label{eq:exact_potentials_compact}
    \begin{bmatrix}
        \bm A(\bm{\psi})\\
        \bm F(\bm{\psi})
    \end{bmatrix}
    = \iint_{S} 
    \begin{bmatrix}
        \mu\,\bm{J}_s(\bm p)\\
        \varepsilon\,\bm{M}_s(\bm p)
    \end{bmatrix} 
    G(\bm{\psi},\bm p) \,dS.
\end{equation}
By eliminating the scalar potentials in \eqref{eq:E_pot_def}, the total electric field can be expressed solely in terms of the vector potentials:
\begin{flalign} \label{eq:E_operator_form}
& \bm E(\bm \psi) = \frac{1}{j\omega\varepsilon}\nabla\times\nabla\times\bm A(\bm{\psi}) - \frac{1}{\varepsilon}\nabla\times\bm F(\bm{\psi}) & \nonumber \\
& \overset{(a)}{=} j\varrho\frac{e^{-j\varrho\|\bm{\psi}\|}}{4\pi \|\bm{\psi}\|}\iint_{S} \frac{\bm{\psi}}{\|\bm{\psi}\|}\times \left[ \bm{M}_s - \eta_0 \big( \bm{J}_s\times\bm{\psi} \big) \right] e^{j\varrho\frac{\bm{\psi}}{\|\bm{\psi}\|} \cdot \bm{p}} \,dS & \nonumber \\
& \overset{(b)}{=} -j\varrho\frac{e^{-j\varrho\|\bm{\psi}\|}}{4\pi \|\bm{\psi}\|} \Bigg\{ \iint_{S} \frac{\bm{\psi}}{\|\bm{\psi}\|}\times (\bm{n}\times \bm{E}^{(\text{PA})}) \, e^{j\varrho\frac{\bm{\psi}}{\|\bm{\psi}\|}\cdot \bm{p}} \,dS & \nonumber \\
& + \eta_0 \iint_{S} \frac{\bm{\psi}}{\|\bm{\psi}\|} \times \big[ (\bm{n}\times \bm{\mathcal H}^{(\text{PA})})\times \frac{\bm{\psi}}{\|\bm{\psi}\|} \big] \, e^{j\varrho\frac{\bm{\psi}}{\|\bm{\psi}\|}\cdot \bm{p}} \,dS \Bigg\} & \nonumber \\
& \overset{(c)}{=} \frac{j \varrho a b \omega \mu s_{m,q}  }{2\varrho_{q}^2 \pi \sqrt{N}  } \frac{e^{-j\varrho\|\bm{\psi}-\bm{\psi}^{(\text{PA})}\|}}{\|\bm{\psi}-\bm{\psi}^{(\text{PA})}\|}\, e^{-\frac{1}{2}\alpha_\text{WG}x^{\text{(PA)}}} & \nonumber \\
& \cdot \frac{\sin\left(\frac{\pi b}{\lambda}\sin\theta\cos\phi\right)}{\frac{\pi b}{\lambda}\sin\theta\cos\phi} \, \frac{\cos\left(\frac{\pi a}{\lambda}\sin\theta\sin\phi\right)}{1 - \left(\frac{2a}{\lambda}\sin\theta\sin\phi\right)^2} & \nonumber \\
&  \cdot \bigg[ \left(1 + \frac{\beta_1}{\varrho}\cos\theta\right)\cos\phi\,\bar{\bm{\vartheta}} + \left(\frac{\beta_1}{\varrho} + \cos\theta\right)\sin\phi\,\bar{\bm{\varphi}} \bigg] & \nonumber \\
& \overset{(d)}{=}\frac{j \varrho a b \omega \mu s_{m,q}  }{2\varrho_{q}^2 \pi \sqrt{N}  }  \frac{e^{-j\varrho\|\bm{\psi}-\bm{\psi}^{(\text{PA})}\|}}{ \|\bm{\psi}-\bm{\psi}^{(\text{PA})}\|}\, e^{-\frac{1}{2}\alpha_\text{WG}x^{\text{(PA)}}} & \nonumber \\
& \cdot \frac{\sin\left(\frac{\pi a}{\lambda}\sin\theta\sin\phi\right)}{\frac{\pi a}{\lambda}\sin\theta\sin\phi} \, \frac{\cos\left(\frac{\pi b}{\lambda}\sin\theta\cos\phi\right)}{1 - \left(\frac{2b}{\lambda}\sin\theta\cos\phi\right)^2} & \nonumber \\
& \cdot \bigg[ \left(1 + \frac{\beta_2}{\varrho}\cos\theta\right)\sin\phi\,\bar{\bm{\vartheta}} + \left(\frac{\beta_2}{\varrho} + \cos\theta\right)\cos\phi\,\bar{\bm{\varphi}} \bigg] &
\end{flalign}
(a) follows from the Fraunhofer condition, under which the standard far-field approximation is applied. (b) follows from the constitutive relations between the equivalent surface currents and the tangential aperture fields in \eqref{eq:J_s} and \eqref{eq:M_s}. (c) follows by substituting \eqref{eq:E_PA^t} and evaluating the resulting two Fourier-type surface integrals in closed form. (d) corresponds to the radiated electric field associated with the mode $q=2$.

\section{Proof of Lemma~\ref{thm:single_angle}}\label{sec:AppB}

Setting $\frac{\partial}{\partial {\delta}^{(\text {PA})}} \ln |H|^2 = 0$ yields
\begin{equation*}
\begin{split}
\frac{2\bar{x}^{\text{(UR)}}}{\| \bar{\bm{\psi}}^{\text{(UR)}}\| - \bar{z}^{\text{(UR)}}} 
&+ \frac{2\bar{z}^{\text{(UR)}}}{\lambda \| \bar{\bm{\psi}}^{\text{(UR)}}\|} \\
&\cdot \left[ \pi b \cot\left( \frac{\pi b \bar{x}^{\text{(UR)}}}{\lambda \| \bar{\bm{\psi}}^{\text{(UR)}}\|} \right) - \frac{\lambda \| \bar{\bm{\psi}}^{\text{(UR)}}\|}{\bar{x}^{\text{(UR)}}} \right] = 0.
\end{split}
\end{equation*}
Since $\bar{x}^{(\text{UR})}$, $\bar{y}^{(\text{UR})}$, and $\bar{z}^{(\text{UR})}$ are trigonometric functions of the rotation angles $(\delta^{(\text {PA})}, \xi^{(\text{PA})})$, the derived equation is transcendental and admits infinite solutions. However, considering that the radiated energy is predominantly concentrated in the main lobe, we restrict our analysis to the vicinity of the main lobe center at $\bar{x}^{(\text{UR})} = 0$. In this region, we apply Taylor expansion to the cotangent term.
\begin{equation*}
\begin{split}
\frac{2\bar{x}^{\text{(UR)}}}{\| \bar{\bm{\psi}}^{\text{(UR)}}\| - \bar{z}^{\text{(UR)}}} 
&- \frac{2\bar{z}^{\text{(UR)}}}{\lambda \| \bar{\bm{\psi}}^{\text{(UR)}}\|} \cdot \frac{(\pi b)^2 \bar{x}^{\text{(UR)}}}{3\lambda \| \bar{\bm{\psi}}^{\text{(UR)}}\|} = 0.
\end{split}
\end{equation*}
The above derivation employs the Taylor expansion of the cotangent function, given by $\cot x \approx \frac{1}{x} - \frac{x}{3} + O(x^3)$, and the expression can be simplified to the following form
\begin{equation*} 
\bar{x}^{\text{(UR)}} \left[- \frac{2 \pi^2 b^2 \bar{z}^{\text{(UR)}}}{3 \lambda^2 \| \bar{\bm{\psi}}^{\text{(UR)}}\|^2}+\frac{1}{\| \bar{\bm{\psi}}^{\text{(UR)}}\|- \bar{z}^{\text{(UR)}}}\right] = 0.
\end{equation*}
Since the user is located below the PA, we have $\bar{z}^{\text{(UR)}} < 0$, which ensures that the first product term in the preceding expression remains strictly positive. Consequently, the function $\ln|H|^2$ attains its extremum at $(\bar{x}^{\text{(UR)}})^* = 0$.

Similarly, for the second equation $\frac{\partial}{\partial {\xi^{(\text{PA})}}} \ln |H|^2 = 0$,   we get $\left( {\bar{y}^{(\text{UR})}} \right)^* = 0$. Substituting the specific expression of $\left( {\bar{y}^{(\text{UR})}} \right)^*$ in equation \eqref{coord_transform} into $\left( {\bar{y}^{(\text{UR})}} \right)^* = 0$, we have
\begin{equation*}
 \left( y^{(\text{UR})} - y^{(\text{PA})} \right)\cos \xi^{(\text{PA})}  + \left( z^{(\text{UR})} - z^{(\text{PA})} \right)\sin \xi^{(\text{PA})}  = 0.
\end{equation*}
Thus, we can solve the optimal $\xi^{(PA)^*}$.

Combining the explicit expression of $\left( {\bar{x}^{(\text{UR})}} \right)^*$ in equation \eqref{coord_transform}, the condition $\left( {\bar{x}^{(\text{UR})}} \right)^* = 0$, and \eqref{e:xi}, the optimal value of $\delta^{(\text{PA})}$ can be derived as 
\begin{equation*}
{\delta^{(\text{PA})}}^* = \arctan\left( \frac{ x^{(\text{UR})} - x^{(\text{PA})} }{ \sqrt{ \left( y^{(\text{UR})} - y^{(\text{PA})} \right)^2 + \left( z^{(\text{UR})} - z^{(\text{PA})} \right)^2 } } \right).
\end{equation*}

Given $\left( {\bar{x}^{(\text{UR})}} \right)^* = \left( {\bar{y}^{(\text{UR})}} \right)^* = 0$, the optimal $\left( {\bar{z}^{(\text{UR})}} \right)^*$ is given by
\begin{eqnarray}
(\bar{z}^{(\text{UR})})^* = -\|\bm{\psi}^{(\text{UR})} - \bm{\psi}^{(\text{PA})}\|,
\end{eqnarray}

To verify that this stationary point is a maximum, we evaluate the Hessian matrix $\bm{H}$ of $\ln |H|^2$ at the optimal solution $({\delta^{(\text{PA})}}^*, {\xi^{(\text{PA})}}^*)$. The evaluation yields a diagonal matrix:
\begin{eqnarray}\label{eq:hessian_matrix}
    &&\hspace{-0.8cm}\bm{H}^{\star} \triangleq 
\begin{bmatrix}
\frac{\partial^2 \ln |H|^2}{\partial (\delta^{(\text {PA})})^2} & \frac{\partial^2 \ln |H|^2}{\partial \delta^{(\text {PA})} \partial \xi^{(\text {PA})}} \\
\frac{\partial^2 \ln |H|^2}{\partial \xi^{(\text {PA})} \partial \delta^{(\text {PA})}} & \frac{\partial^2 \ln |H|^2}{\partial (\xi^{(\text {PA})})^2}
\end{bmatrix}_{({\delta^{(\text {PA})}}^*, {\xi^{(\text {PA})}}^*)} \\
    &&\hspace{-0.8cm} =
\begin{bmatrix}
-\frac{2(\pi b)^2}{3\lambda^2} - 1 & 0 \\
0 & \frac{\left( y^{(\text{UR})} \right)^2 + \left( z^{(\text{UR})} \right)^2}{\| \bar{\bm{\psi}}^{\text{(UR)}}\|^2} \left[ \frac{2 a^2 (8 - \pi^2)}{\lambda^2} - 1 \right] \notag
\end{bmatrix}.
\end{eqnarray}
Given that the antenna dimensions ensure the first diagonal term is negative, and noting that $(8 - \pi^2) < 0$ and $\frac{\left( y^{(\text{UR})} \right)^2 + \left( z^{(\text{UR})} \right)^2}{\| \bar{\bm{\psi}}^{\text{(UR)}}\|^2}  \ge 0$, we conclude that the second diagonal term is inherently negative. Consequently, the Hessian matrix $\bm{H}^{\star}$ is negative definite, confirming that the solution corresponds to a strict local maximum.

Moreover, the terms in $|H|^2$ that depend on $(\theta,\phi)$ can be written as
\begin{eqnarray*}
    &&\hspace{-0.8cm}|H|^2 \propto (1+\cos\theta)^2 \cdot
\mathrm{sinc}^2\!\left(\sin\theta \cdot\cos\phi\cdot\frac{b}{\lambda}\right)\\
 &&\hspace{-0.5cm}\cdot \frac{\cos^2 \left(\sin\theta \cdot\sin\phi\cdot\frac{a}{\lambda}\pi\right)}{\big[1-\left(2\sin\theta \cdot\sin\phi\cdot\frac{a}{\lambda}\pi\right)^2\big]^2} .
\end{eqnarray*}
The obliquity factor satisfies $(1+\cos\theta)^2\le 4$ with equality at $\theta=0$.
The squared $\mathrm{sinc}$ term is upper bounded by $1$ with equality when
$\sin\theta\cos\phi=0$.
In the main-lobe neighborhood, the remaining cosine-distribution term is maximized at
$\sin\theta\sin\phi=0$.
Consequently, all angle-dependent factors attain their theoretical upper bounds strictly under the condition $(\bar{x}^{(\text{UR})})^* = (\bar{y}^{(\text{UR})})^* = 0$ with $(\bar{z}^{(\text{UR})})^* = -\|\bm{\psi}^{(\text{UR})} - \bm{\psi}^{(\text{PA})}\|$. Thus, $|H|^2$ achieves its global maximum at this specific configuration.

\section{Proof of Lemma~\ref{thm:optimal_x}}\label{sec:AppC}
To find the optimal PA position $x^{(\text{PA})}$,  we define a new function to analyze the trade-off between waveguide loss and atmospheric absorption loss in \eqref{e:partial_ln_H_2}, we have
\begin{equation}\label{eq:app_foc}
F(d)\triangleq -\alpha_{\text{W}}
+\alpha_{\text{A}}\frac{d}{\sqrt{d^2+\rho^2}}
+\frac{2d}{d^2+\rho^2}
=0,
\end{equation}
where $d \triangleq x^{(\text{UR})} - x^{(\text{PA})}$ represents the horizontal distance to the user, and $\rho \triangleq \sqrt{ (y^{(\text{UR})} - y^{(\text{PA})})^2 + (z^{(\text{UR})} - z^{(\text{PA})})^2 }$ is the constant radial distance in the transverse plane. 

The derivative  of $F(d)$ with respect to $d$ is
\begin{equation*}\label{eq:app_Fprime}
\frac{\partial F(d)}{\partial d}
=
\alpha_{\text{A}}\frac{\rho^2}{(d^2+\rho^2)^{3/2}}
+\frac{2(\rho^2-d^2)}{(d^2+\rho^2)^2}.
\end{equation*}
The function $F(d)$ attains an extremum where its derivative is zero. Since $\sqrt{d^2+\rho^2}$ is always positive, the above equation simplifies to
\begin{equation*}
\alpha_{\text{A}}^2\rho^4(d^2+\rho^2)+4(\rho^2-d^2)^2=0
\end{equation*}

For such a transcendental equation, Ferrari's method typically yields four roots. However, two of these are extraneous roots that fall outside the domain of definition. Upon simplification, the two valid extrema of the function $F(d)$ are obtained as follows:
\begin{equation*}
\begin{split}
    d_1 = \sqrt{\frac{8\rho^2 + \alpha_{\text{A}}^2 \rho^4 + \alpha_{\text{A}} \rho^4\sqrt{\alpha_{\text{A}}^2  + 16(\rho^2+1)}}{8}},\\
d_2 = \sqrt{\frac{8\rho^2 + \alpha_{\text{A}}^2 \rho^4 - \alpha_{\text{A}} \rho^4\sqrt{\alpha_{\text{A}}^2  + 16(\rho^2+1)}}{8}}.
\end{split}
\end{equation*}

The second derivative of the function $F(d)$ with respect to the variable $d$ can be expressed as:
\begin{equation*}
    \frac{\partial^2 F(d)}{\partial d^2}=\frac{d\left(4(d^2-3\rho^2)\sqrt{d^2+\rho^2}-3\alpha_{\text{A}}\rho^2(d^2+\rho^2)\right)}{(d^2+\rho^2)^\frac{7}{2}}
\end{equation*}
Substituting $d_1$ and $d_2$ with $\alpha_{\text{A}}=0.05$ and $\alpha_{\text{W}}=0.08$ for the case $\rho \ge 3$ yields:
\begin{equation*}
    \left. \frac{\partial^2 F}{\partial d^2} \right|_{d=d_2}=\frac{64}{\rho^3}\cdot\frac{\sqrt{Q-8}}{Q^3}\left(\sqrt{2}(Q-32)-3\alpha_{\text{A}}\rho\sqrt{Q}\right)<0
\end{equation*}
where $Q \triangleq 16+\alpha_{\text{A}}^2\rho^2-\alpha_{\text{A}}\rho^2\sqrt{\alpha_{\text{A}}^2+16(\rho^2+1)}$. It is evident that the maximum is located at $d_2$, which is closer to the user. Substituting the specific expression for $d_1$ and $d_2$, we can observe that $F(d_1)$ and $F(d_2)$ are all positive.

Given that $F(0) < 0$ corresponds to the scenario where the antenna is vertically aligned with the user, it follows that $|H|^2$ attains a maximum with respect to $x^{(PA)}$ within the interval bounded by $d_2$ and $d=0$.

We introduce an intermediate variable into equation \ref{eq:app_foc} by setting $d=\frac{\rho \cos \omega}{\sin \omega}>0$ for $\omega \in (0, \pi/2)$. Then $d^2 = \frac{\rho^2 \cos^2 \omega}{\sin^2 \omega}$, yielding
\begin{equation}
F(d)= -\alpha_{\text{W}}+\alpha_{\text{A}}\cos \omega+\frac{1}{\rho}\sin(2\omega)
\end{equation}
By setting $F(d)=0$ and applying the Taylor approximation $\omega^*= \frac{\pi}{2} - \frac{\alpha_\text{W}}{\alpha_{\text{A}}+\frac{\rho}{2}} $,  we obtain:
\begin{equation}
d^*= \frac{\rho^2\alpha_{\text{W}}}{\rho\alpha_{A}+2}
\end{equation}

To prove Lemma~\ref{thm:optimal_x}, we have to analyze the trade-off between the two user locations. For a single multimode pinching antenna serving two users, the derivative of the sum rate can be expressed as follows:
\begin{equation}\label{eq:partial_R_sum}
\frac{\partial R_s(x)}{\partial x}
=\frac{1}{\ln 2}\sum_{k=1}^{2}\frac{p_k|H_k(x)|^2}{\sigma^2_k+p_k|H_k(x)|^2}
\cdot
\frac{\partial \ln |H_k(x)|^2}{\partial x}.
\end{equation}
where $\frac{p_k|H_k(x)|^2}{\sigma^2_k+p_k|H_k(x)|^2}\in (0,1)$ captures the marginal rate gain of user $k$ at its current SNR level, and therefore naturally governs the trade-off in the shared placement variable $x$. Based on the preceding conclusion, it is straightforward to verify that $R_s' > 0$ when $x < x_1^\star$ and $R_s' < 0$ when $x > x_2^\star$. Consequently, the optimal $x$ that maximizes the sum rate resides within the interval $[x_1^\star, x_2^\star]$. This completes the proof.

\section{Proof of  Theorem~\ref{lem:optimal_position}}\label{sec:AppD}

To establish Theorem~\ref{lem:optimal_position}, we characterize the local behavior of $R_q(x^{(\mathrm{PA})})$ around the single-user optimum $x_q^*$. In the two-user setting, the optimal PA location deviates from $x_q^*$ due to the spatial compromise of sharing a common antenna position and the resulting position-dependent power coupling. Evaluating the derivatives at $x_q^*$ analytically quantifies this joint perturbation. Differentiating $R_q(x^{(\mathrm{PA})})$ yields
\begin{eqnarray}
&&\hspace{-1cm} R_q^{\prime}\bigg|_{x^{(\text{PA})}=x_q^*}
\!\!\!\!\!\!\!\!=\!\frac{\!|w^{(p)}_q(x^{(\text{PA})})|^2\!\frac{\partial |h_q(x^{(\mathrm{PA})})|^2}{\partial x^{(\text{PA})}}\!+\!|h_q(x^{(\text{PA})})|^2\!\frac{\partial |w^{(p)}_q(x^{(\text{PA})})|^2}{\partial x^{(\text{PA})}}\!}{\ln2(\frac{\sigma_q^2}{P}+|w^{(p)}_q(x^{(\text{PA})})|^2|h_q(x^{(\text{PA})})|^2)}\notag \\
&&\hspace{-0cm} \overset{(a)}{=}\frac{P|h_q(x_q^*)|^2\frac{\partial |w^{(p)}_q(x^{(\text{PA})})|^2}{\partial x^{(\text{PA})}} }{\ln2(\sigma_q^2+P|w^{(p)}_q(x_q^*)|^2|h_q(x_q^*)|^2)}\notag \\
&&\hspace{-0.3cm}=-\frac{\sigma_2^2}{2\ln 2}\,\frac{\dfrac{1}{|h_{q^{\prime}}(x_q^*)|^2}\left.\dfrac{\partial \ln |h_{q^{\prime}}(x^{(\text{PA})})|^2}{\partial x^{(\text{PA})}}\right|_{x^{(\text{PA})}=x_q^*}}
{\dfrac{\sigma_1^2}{|h_q(x_q^*)|^2}+P|w^{(p)}_q(x_q^*)|^2},
\end{eqnarray}
where step (a) follows by substituting $x^{(\text{PA})}=x_q^*$. Noting that $x_q^*$ is the unconstrained local maximizer for user $q$, and hence $\frac{\partial |h_q(x^{(\text{PA})})|^2}{\partial x^{(\text{PA})}}\big|_{x^{(\text{PA})}=x_q^*} = 0$. Differentiating once more yields the exact second-order expression
\begin{eqnarray}
&&\hspace{-0.8cm} R_q^{\prime\prime}(x_q^*)=
\frac{P}{\ln 2 (\sigma_q^2 + P|h_q(x_q^*)|^2|w^{(p)}_q(x_q^*)|^2)} \Bigg[ \bigg(|w^{(p)}_q(x_q^*)|^2\notag \\
&&\hspace{-1cm} +\frac{\sigma_q^2}{2P|h_q(x_q^*)|^2}\bigg) \frac{\partial^2|h_q(x_q^*)|^2}{\partial (x_q^*)^2}-\frac{\sigma_{q^\prime}^2 |h_q(x_q^*)|^2}{2P|h_{q^\prime}(x_q^*)|^4}\frac{\partial^2|h_{q^\prime}(x_q^*)|^2}{\partial (x_q^*)^2} \notag\\
&&\hspace{-1cm} +\bigg(\frac{\partial|h_{q^\prime}(x_q^*)|^2}{\partial x_q^*} \bigg)^2\bigg(\frac{\sigma_{q^\prime}^2|h_q(x_q^*)|^2}{P|h_{q^\prime}(x_q^*)|^6}-\frac{\sigma_{q^\prime}^4|h_q(x_q^*)|^4}{P|h_{q^\prime}(x_q^*)|^8} \notag \\
&&\hspace{-0.5cm}\cdot\frac{1}{4(\sigma_q^2 + P|h_q(x_q^*)|^2|w^{(p)}_q(x_q^*)|^2)}\bigg) \Bigg],
\end{eqnarray}
Although exact, the above expression is too cumbersome to provide direct insight. The next three steps extract the dominant curvature terms under the weak-coupling regime of interest. In particular, we aim to identify which part of the curvature is caused by the residual direct variation of user $q$ itself, and which part is caused by the interaction with the alternate mode $q'$.
\begin{eqnarray}
&&\hspace{-0.8cm} R_q^{\prime\prime}(x_q^*)\overset{(a)}{\approx}\frac{P}{\ln 2 (\sigma_q^2 + P|h_q(x_q^*)|^2|w^{(p)}_q(x_q^*)|^2)}\Bigg[\frac{\partial^2|h_{q^\prime}(x_q^*)|^2}{\partial (x_q^*)^2}\notag \\
&&\hspace{-1cm}\cdot\bigg(|w^{(p)}_q(x_q^*)|^2\!\!-\!\!\frac{\sigma_{q^\prime}^2 |h_q(x_q^*)|^2}{2P|h_{q^\prime}(x_q^*)|^4}\bigg)\!\!+\!\!\bigg(\!\frac{\partial|h_{q^\prime}(x_q^*)|^2}{\partial x_q^*} \!\bigg)^2\!\!\bigg(\frac{\sigma_{q^\prime}^2|h_q(x_q^*)|^2}{P|h_{q^\prime}(x_q^*)|^6}
\notag \\
&&\hspace{-0.8cm}-\frac{\sigma_{q^\prime}^4|h_q(x_q^*)|^4}{P|h_{q^\prime}(x_q^*)|^8}\cdot\frac{1}{4(\sigma_q^2 + P|h_q(x_q^*)|^2|w^{(p)}_q(x_q^*)|^2)}\bigg) \Bigg] \notag \\
&&\hspace{-0.8cm} \overset{(b)}{\approx}R_q^\prime(x_q^*)\bigg(\frac{\frac{\partial^2|h_{q^\prime}(x_q^*)|^2}{\partial (x_q^*)^2} }{\frac{\partial|h_{q^\prime}(x_q^*)|^2}{\partial x_q^*} } -\frac{\frac{\partial|h_{q^\prime}(x_q^*)|^2}{\partial x_q^*} }{|h_{q^\prime}(x_q^*)|^2}\bigg)-\ln2\bigg(R_q^\prime(x_q^*)\bigg)^2\notag \\
&&\hspace{-0.8cm}\overset{(c)}{\approx}- \ln 2 \! \left(\! R_q^{\prime}(x_q^*) \! \right)^{\!2} \!-\! R_q^{\prime}(x_q^*) \left.  \!\!\frac{\partial \! \ln \! |h_{q^{\prime}}(x^{(\text{PA})}\!) \! |^{2}}{\partial x^{(\text{PA})}} \!\right|_{ x^{(\text{PA})}=x_q^*},
\end{eqnarray} 
where step (a) follows from two approximations. First, the second-order derivatives of $|h_q(x_q^*)|^2$ and $|h_{q^\prime}(x_q^*)|^2$ with respect to $x_q^*$ are treated as being of the same order, so their difference is neglected at this stage. Second, the coefficient associated with the $q'$-related term is much larger than that associated with the $q$-related term, 
$\frac{\sigma_{q'}^2 |h_q(x_q^*)|^2}{2P|h_{q'}(x_q^*)|^4}
\gg \frac{\sigma_q^2}{2P|h_q(x_q^*)|^2}.$
Therefore, the contribution of the latter is ignored. 
Step (b) follows from the fact that, in the considered regime, $\frac{\sigma_{q'}^2 |h_q(x_q^*)|^2}{2P|h_{q'}(x_q^*)|^4} \gg |w^{(p)}_q(x_q^*)|^2.$
Step (c) follows because, at $x^{(\mathrm{PA})}=x_q^*$, the ratio between the first derivative of $|h_{q'}(x^{(\text{PA})})|^2$ and its value is larger than the ratio between its second derivative and its first derivative. 



\bibliographystyle{IEEEtran}
\bibliography{Ref}

\end{document}